\documentclass[a4paper,11pt,leqno]{amsart}
\usepackage{a4wide,color,array,graphicx}
\usepackage{mathtools}
\usepackage{hyperref}
\usepackage{cite}
\hypersetup{linktocpage,colorlinks, citecolor={magenta}}
\usepackage{amsmath,amssymb,amsthm}
\usepackage[english]{babel}
\usepackage[utf8]{inputenc}
\usepackage[cyr]{aeguill}
\usepackage{esint}
\usepackage{xargs}
\usepackage{lmodern}
\usepackage{bm}

\theoremstyle{plain}

\newtheorem{theo}{Theorem}
\newtheorem{prop}{Proposition}[section]
\newtheorem{lem}[prop]{Lemma}
\newtheorem{coro}[prop]{Corollary}

\theoremstyle{definition}
\newtheorem{remark}[prop]{Remark}
\newtheorem{defi}[prop]{Definition}
\newtheorem{example}[prop]{Example}

\numberwithin{equation}{section}

\newcommand{\R}{\mathbb{R}}
\def\t0{\rightarrow 0} 
\def\ti{\rightarrow \infini} 

\newcommand{\f}{\frac}
\newcommand{\infini}{\infty}

\newcommand{\hal}{\frac{1}{2}}

\def\div{\mathrm{div} \, } 
\def\1{\mathbf{1}} 

\def \mc{\mathcal}

\renewcommand{\epsilon}{\varepsilon}

\def\meseq{\mu_{V}} 
\def \ZNbeta{Z_{N,\beta}} 

\def\T{\mathbb{T}}
\def \um{\underline{m}} 
\def\({\left(}
\def\){\right)}

\def \W{\mathbb{W}} 
\def \bttW{\overline{\W}} 

\def\config{\mathrm{Config}} 
\def\Elec{\mathrm{Elec}}
\def \probas{\mathcal{P}}

\def\P{\mathbb{P}} 
\def \bPst{\bar{P}} 
\def \Pgot{\mathfrak{P}}
\def \bPgot{\overline{\Pgot}}
\def \PNbeta{\P_{N, \beta}} 
\def \PgN2{\mathbf{P}_{N,2}} 

\def \QNbeta{\mathbb{Q}_{N,\beta}}

\def \bQpN{\bar{\mathfrak{Q}}_{N,\beta}} 
\def \sineb{\mathrm{Sine}_{\beta}} 

\def \HN{\mathcal{H}_N}

\def\Esp{\mathbf{E}} 
\def \E{\Esp}

\def \ERS{\mathsf{ent}} 
\def \bERS{\overline{\mathsf{ent}}} 

\def \Poisson{{\Pi}}

\def \fbarbeta{\overline{\mathcal{F}}_{\beta}} 

\def\um{\underline{m}}

\def\Ani{A}
\def\Vt{V_t}
\def\zetat{\zeta_t}

\def\muv{\meseq}
\def\mueq{\meseq}
\def\I{\mathcal I}

\def\nab{\nabla}

\def\indic{\mathbf{1}}

\def\id{\mathrm{Id}}

\def\XXint#1#2#3{{\setbox0=\hbox{$#1{#2#3}{\int}$}
     \vcenter{\hbox{$#2#3$}}\kern-.5\wd0}}

\def \carr{\square} 

\def  \comeg{c_{\omega,\Sigma}}

\def \Ctot{\mathcal{C}^{\rm{tot}}}

\def \Etot{E^{\rm{tot}}}

\def\P{\mathbb{P}} 
\def \bPst{\bar{P}} 
\def \Pgot{\mathfrak{P}}
\def \bPgot{\overline{\Pgot}}
\def \PNbeta{\P_{N, \beta}} 
\def \PgN2{\mathbf{P}_{N,2}} 

\def \bQpN{\bar{\mathfrak{Q}}_{N,\beta}} 
\def \sineb{\mathrm{Sine}_{\beta}} 

\def\g{\mathsf{g}}
\def \V{V}
\def \I{\mathcal{I}}
\def \LogU{\textsf{Log1}}
\def \LogD{\textsf{Log2}}
\def\Coul{\textsf{Coul}}
\def \Riesz{\textsf{Riesz}}
\def \emp{\mu^{\mathrm{emp}}_N}
\def \XN{\vec{X}_N}
\def\YN{\vec{Y}_N}
\def \Emp{\mathrm{Emp}}
\def \bEmp{\overline{\Emp}}
\def \C{\mathcal{C}}
\def \bP{\bar{P}}
\def \KNbeta{K_{N, \beta}}

\def\nab{\nabla}

\def\pa{{\partial}}

\def\ro{\rho}

\def \fluct{\mathrm{fluct}}
\def\Fluct{\mathrm{Fluct}}

\def\hal{\frac{1}{2}}

\def\EK{E^{(K)}}

\def\I{{\mathcal{I}}}
\makeatletter
\def\namedlabel#1#2{\begingroup
    #2%
    \def\@currentlabel{#2}%
    \phantomsection\label{#1}\endgroup
}
\makeatother

\def \d{\mathsf{d}}
\def \s{\mathsf{s}}

\def \f{\mathsf{f}}

\def \hpN{H'_{N}}
\def \hpNe{H'_{N, \eta}}

\def\Rd{\R^\d} 

\def \drd{\delta_{\Rd}}
\def \ent{\mathrm{ent}}
\def \bPstx{\bar{P}^{x}}

\def\cd{\mathsf{c}_{\d}}

\def\Esp{\mathbf{E}} 
\def \E{\Esp}

\def \bm{\begin{displaymath}}
\def \em{\end{displaymath}}
\def \be{\begin{equation}}
\def \ee{\end{equation}}
\def \beq*{\begin{equation*}}
\def \eeq*{\end{equation*}}
\def \ba{\begin{eqnarray}}
\def \ea{\end{eqnarray}}
\def \ba*{\begin{eqnarray*}}
\def \ea*{\end{eqnarray*}}
\def\mueqt{\mu_{V_t}}

\def\zetat{\zeta_t}

\long\def\replace#1{#1}

\replace{%

}

\begin{document}

%
%
%
%
%
%

\title{Microscopic description of Log and Coulomb gases}

%
%
\author{Sylvia Serfaty}
\address{Courant Institute of Mathematical Sciences, New York University, 251 Mercer st, New York, NY 10012}
\email{serfaty@cims.nyu.edu}
%
%

\begin{abstract}
These are the lecture notes of a course taught at the Park City Mathematics Institute in June 2017. They are intended to review some recent results, obtained in large part with Thomas Lebl\'e, on the statistical mechanics of systems of points with logarithmic or Coulomb interactions. After listing some  motivations, we describe the ``electric approach"  which allows to get  concentration results, Central Limit Theorems for fluctuations, and a Large Deviations Principle expressed in terms of the microscopic state of the system.
\end{abstract}  

%
%
\maketitle
\thispagestyle{empty}

%
%


\section{Introduction and motivations}\label{intro}
We are interested in the following class of energies
\begin{equation} \label{HN}
\HN(x_1, \dots, x_N) := \sum_{1 \leq i \neq j \leq N} \g(x_i-x_j) +  \sum_{i=1}^N N \V(x_i).
\end{equation}
where $x_1, \dots , x_N$ are  $N$ points (or particles) in the Euclidean space $\R^\d$ ($\d \geq 1$), and $N$ is large.  
The potential $V$ is a confining potential, growing fast enough at infinity, on which we shall make assumptions later.
The pair interaction potential $\g$ is given by either
\begin{align} 
\label{wlog} & (\LogU \ \text{case}) \quad \g(x) = -\log |x| , \quad \text{in dimension } \d=1, \\
\label{wlog2d} & (\LogD \ \text{case}) \quad \g(x) = - \log |x| , \quad  \text{in dimension } \d=2, \\
\label{kernel} & (\Coul \ \text{case}) \quad \g(x) = |x|^{2-\d} \text{ in dimension $\d \geq 3$}.
\end{align}

We will also say a few things about the more general case
\begin{equation}\label{kernel2}  (\Riesz \ \text{case}) \quad \g(x) = |x|^{-\s}, \quad   \text{with }\max(0, \d-2)\leq \s<\d, \text{ in dimension $\d \geq 1$}.\end{equation}
 The interaction as in \eqref{wlog} (resp. \eqref{wlog2d}) corresponds to a one-dimensional (resp. two-dimensional) logarithmic interaction, we will call \LogU, \LogD \ the \textit{logarithmic cases}. The \LogD \ interaction is also the Coulomb interaction in dimension $2$. For $\d \geq 3$, \Coul \  corresponds to the Coulomb interaction in higher dimension. 
 In both instances \LogD \  and \Coul\  we have 
 \be\label{coulombkernel}
 -\Delta \g= \cd \delta_0\ee where 
 \be\label{defcd}
\cd= 2\pi \quad\text{if} \ \d =2 \qquad \cd= (d-2)|\mathbb{S}^{\d-1}|\quad \text{for }\ \d \ge 3.\ee

 Finally, the \Riesz \ cases $\max(\d-2, 0) < \s < \d$ in \eqref{kernel} correspond to more general Riesz interactions, in what is called the {\it potential case}. The situation where $\s>\d$, i.e. that for which $\g$ is not integrable near $0$ is called the {\it hypersingular case}. Some generalizations of what we discuss here to that case are provided in \cite{hlss}, see references therein for further background and results.
 
   Whenever the parameter $\s$ appears, it will be with the convention that $\s = 0$ in the logarithmic \medskip cases. 

We will use the notational shortcut $\XN$ for $(x_1, \dots, x_N)$ and $d\XN $ for $ dx_1 \dots dx_N$.

For any $\beta > 0$, we consider the canonical Gibbs measure at inverse temperature $\beta$, given by the following density
\begin{equation}\label{gibbs}
d\PNbeta(\XN) = \frac{1}{\ZNbeta} \exp \left( -\frac{\beta}{2}  \HN(\XN)  \right) d\XN
\end{equation}
where 
\begin{equation}
\label{defZ}
\ZNbeta= \int  \exp \left( -\frac{\beta}{2}  \HN(\XN)  \right) d\XN.\ee
The term $\ZNbeta$ is a normalizing constant, called the \textit{partition function}, which plays an important role in understanding the physics of the system.

We now review various motivations for studying such systems.
\subsection{Fekete points and approximation theory}
Fekete points arise in interpolation theory as the points minimizing  interpolation errors for numerical integration \cite{safftotik}.   Indeed, if one is looking for  $N$ interpolation points $\{x_1, \dots, x_N\}$ in $K$ such that the  relation
\[
\int_K f(x) dx = \sum_{j=1}^N w_j f(x_j)
\]
is exact for the polynomials of degree $\le N-1$. One sees that one needs to compute the coefficients $w_j$ such that  $\int_{K} x^k = \sum_{j=1}^N w_j x_j^k$ for $0 \leq k \leq N-1$, and this computation is easy if one knows to invert the Vandermonde matrix of the  $\{x_j\}_{j=1 \dots N}$. The numerical stability of this operation is as large as the  \textit{condition number} of the  matrix, i.e. as the Vandermonde determinant of the $(x_1, \dots, x_N)$. The  points that minimize the maximal interpolation error for general functions are easily shown to  be the Fekete points, defined as those that maximize
$$\prod_{i\neq j} |x_i-x_j|$$
 or equivalently minimize
$$-\sum_{i\neq j} \log |x_i-x_j|.$$

They are often studied on manifolds, such as the $\d$-dimensional sphere.
In Euclidean space, one also considers ``weighted Fekete points" which maximize
$$\prod_{i\neq j} |x_i-x_j| e^{-N\sum_i V(x_i)}$$
or equivalently minimize
$$-\sum_{i\neq j} \log |x_i-x_j| + N\sum_{i=1}^N V(x_i)$$
which in dimension $2$ corresponds exactly to the minimization of our Hamiltonian $\HN$ in the particular case \LogD.  They also happen to be zeroes of orthogonal polynomials, see \cite{simon}.

Since $-\log |x|$ can be obtained as $\lim_{s\to 0} \frac{1}{s}(|x|^{-s}-1)$, there is also interest in studying ``Riesz $\s$-energies", i.e. the minimization of 
\begin{equation}\label{rieszs}
\sum_{i\neq j} \frac{1}{|x_i-x_j|^\s}\end{equation}
for all possible $\s$, hence a motivation for \eqref{kernel}. Varying $\s$ from $0$ to $\infty$ connects  \textit{Fekete points} to optimal packing problems. The optimal packing problem is solved in 1, 2 and 3 dimensions. The solution in dimension 2 is the triangular lattice, in dimension 3 it is the FCC (face-centered cubic) lattice. In higher dimension, the solution is in general not known except in dimensions 8 and 24. In high dimension, where the problem is important for error-correcting codes,  it is expected that it is {\it not} a lattice.

Note that the triangular lattice is conjectured to have to universally minimizing property \cite{ck} i.e. to be the minimizer for a broad class of interactions. An analogous role is played in dimensions 8 and 24 by the $E_8$ and Leech lattices, respectively, and remarkably the universally minimizing property of these two lattices has just been recently proven \cite{via,ckrmv}.

For these aspects, we refer to the  the review papers \cite{sk,bhs}  and references therein (in that context such systems are mostly studied on the $\d$-dimensional sphere or torus).

\subsection{Statistical mechanics}

The ensemble given by \eqref{gibbs} in the  \LogD \ case is called in physics a two-dimensional \textit{Coulomb gas} or \textit{one-component plasma} (see e.g. \cite{aj,jlm,janco,sm} for a physical treatment). Two-dimensional Coulomb interactions arise in quantum mechanics: the density of the many-body wave function of  free fermions in a harmonic trap  is the same as the Gibbs measure of the \LogU \ gas with $\beta=2$ \cite{ddms}, and the fractional quantum Hall effect is also described via a two-dimensional Coulomb gas \cite{Gir,stormer}.   Ginzburg-Landau vortices \cite{livre} and vortices in superfluids and Bose-Einstein condensates also interact like two-dimensional Coulomb particles, cf. below.  The \Coul \ case with $d = 3$ can be seen as a toy (classical) model for matter (see e.g. \cite{PenroseSmith,jlm,LiLe1,LN}).

The general \Riesz \ case can be seen as a generalization of the Coulomb case, motivations for studying Riesz gases are numerous in the physics literature (in solid state physics, ferrofluids, elasticity), see for instance \cite{mazars,bbdr,CDR,To}, they can also correspond to systems with Coulomb interaction constrained to a lower-dimensional subspace. 

In all cases of interactions, the systems governed by the Gibbs measure $\PNbeta$ are considered as difficult systems in statistical mechanics because the interactions are truly long-range, singular, and the points are not constrained to live on a lattice.

As always in statistical mechanics \cite{huang}, one would like to understand if there are phase-transitions for particular values of the (inverse) temperature $\beta$. For the systems studied here, one may expect what physicists call a liquid for small $\beta$, and a crystal for large $\beta$. The meaning of crystal in this instance is not to be taken literally as a lattice, but rather as a system of points whose 2-point correlation function
$$\ro_2(x,y)= T_2(x-y)$$ (assuming translation invariance) does not decay too fast as $x-y \to \infty$.

This phase-transition  at finite $\beta$ has been conjectured in the physics literature for the \LogD \ case (see e.g. \cite{bst,caillol1982monte,choquard1983cooperative}) but its precise nature is still unclear (see e.g. \cite{stishov} for a discussion). In view of the recent progress in computational physics concerning such phenomenon in two-dimensional systems (see e.g. \cite{kapfer2015two}), which suggest a possibly very subtle transition between the liquid and solid phase, the question seems yet out of reach for a rigorous treatment.

\subsection{Two component plasma case} 
The two-dimensional ``one component plasma", consisting of positively charged particles, has a ``two-component" counterpart which consists in $N$ particles $x_1, \dots , x_N$ of charge $+1$ and $N$ particles $y_1, \dots , y_N$ of charge $-1$ interacting logarithmically, with energy 
$$\HN(\XN, \YN)= - \sum_{i\neq j} \log |x_i-x_j|- \sum_{i\neq j} \log |y_i-y_j|+  \sum_{i, j} \log |x_i-y_j|$$
and the Gibbs measure $$\frac{1}{\ZNbeta} e^{-\beta \HN(\XN,\YN) } d\XN\, d\YN.$$
Although the energy is unbounded below (positive and negative points attract), the Gibbs measure is well defined for $\beta$ small enough, more precisely the partition function converges for $\beta<2$. The system is then seen to form dipoles which do not collapse, thanks to the thermal agitation. 
The two-component plasma is interesting due to its close relation to two important theoretical physics models: the XY model and the sine-Gordon model (cf. the review \cite{spencer}), which exhibit  a Kosterlitz-Thouless phase transition (made famous by the 2016 physics Nobel prize, cf. \cite{bietenholzgerber}) consisting in the binding of these ``vortex-antivortex" dipoles.

\subsection{Random matrix theory}
The study of \eqref{gibbs} has attracted a lot of attention due to its connection with random matrix theory (we refer to \cite{forrester} for a comprehensive treatment).  Random matrix theory (RMT) is a relatively old theory, pionereed by statisticians and physicists such  as Wishart, Wigner and Dyson, and originally motivated by the understanding of the spectrum of heavy atoms, see \cite{mehta}. For more recent mathematical reference see \cite{agz,deift,forrester}. The main question asked by RMT is~: what is the law of the spectrum of a large random matrix ? 
As first noticed in the foundational papers of  \cite{wigner,dyson}, in the particular cases \eqref{wlog}--\eqref{wlog2d}  the Gibbs measure \eqref{gibbs}  corresponds in some particular instances to the joint law of the eigenvalues (which can be computed algebraically) of some famous random matrix ensembles:
\begin{itemize}

\item for \LogD , $\beta=2$ and $V(x)=|x|^2$, \eqref{gibbs} is the law of the (complex)  eigenvalues of an $N\times N$ matrix where the entries are chosen to be normal Gaussian i.i.d.  This is called the {\it Ginibre ensemble}.

\item for \LogU , $\beta=2$ and $V(x)= x^2/2$, \eqref{gibbs} is the law of the (real) eigenvalues of an $N\times N$ Hermitian matrix with complex normal Gaussian iid entries. This is called the Gaussian Unitary Ensemble.

\item for \LogU  , $\beta=1$ and $V(x)=x^2/2$, \eqref{gibbs} 
 is the law of the (real) eigenvalues of an $N\times N$ real symmetric  matrix with normal Gaussian iid entries. This is called the Gaussian Orthogonal Ensemble.

\item for \LogU ,  $\beta=4$ and $V(x)=x^2/2$, \eqref{gibbs} is the law of the eigenvalues of an $N\times N$  quaternionic symmetric  matrix with normal Gaussian iid entries. This is called the Gaussian Symplectic Ensemble. 

\end{itemize}
 One thus observes in these ensembles the phenomenon of ``repulsion of eigenvalues": they repel each other logarithmically, i.e.  like two-dimensional Coulomb particles.

 For the \LogU \ and \LogD \  cases, at the  specific temperature $\beta=2$, the law \eqref{gibbs} acquires a  special algebraic feature : it becomes a {\it determinantal} process, part of a wider class of processes (see \cite{bkpv,borodin}) for which the correlation functions are explicitly given by certain  determinants. This allows for many explicit  algebraic computations, on which there is a large literature.  In particular, one can compute an expansion of $\log \ZNbeta$ as $N\to \infty$ (see \cite{mehta}) and  the limiting processes at the microscopic scale \cite{bors}.
 However,  many relevant quantities that can be computed explicitly for $\beta = 2$ are not exactly known for the $\beta \neq 2$ case, even in the case of  the potential $V(x) = |x|^2$. 
   The particular case of \eqref{wlog} for all $\beta$ and general $V$, also called $\beta$-ensembles, has however been well-understood. In particular, thanks to the works of
    \cite{joha,shch,bey1,bey2,bfg,bl}, one has expansions of $\log \ZNbeta$, Central Limit Theorems for linear statistics, and   {\it universality} (after suitable rescaling) of the microscopic behavior and local statistics of the points, i.e. the fact that they are essentially independent of $V$.
   
Considering the coincidence between a statistical mechanics model and the law of the spectrum of a random matrix model for several values of the inverse temperature, it is also natural to ask whether such a correspondence exists for any value of $\beta$, i.e. whether \eqref{gibbs} can be seen as a law of eigenvalues for some random matrix model. The answer is positive in dimension 1 for any $\beta$: a somehow complicated model of tridiagonal matrices can be associated to the Gibbs measure of the one-dimensional log gas at inverse temperature $\beta$, see \cite{kn,de}.  This and other methods  allow again  to compute a lot explicitly, and to derive that  the microscopic laws of the eigenvalues are those of a so called {\it sine-$\beta$ process} \cite{vv}.

   \subsection{Complex geometry and theoretical physics}
   Two-dimensional Coulomb systems (in the determinantal case $\beta=2$)  are of interest to geometers   because they serve to construct K\"ahler-Einstein metrics with positive Ricci curvature on complex manifolds, cf. \cite{berman}. 
   Another  important motivation is the construction of Laughlin states for the Fractional Quantum Hall effect, which effectively reduces to the study of a two-dimensional Coulomb gas (cf. \cite{Gir,stormer,rougerieyngvason}). When studying the Fractional Quantum Hall effect on a complex manifold, the coefficients in the expansion of the (logarithm of the) partition function have interpretations  as geometric invariants, and it is thus of interest to be able to compute them, cf \cite{klevtsov}.

\subsection{Vortices in condensed matter physics}
In superconductors with applied magnetic fields, and in rotating  superfluids and Bose-Einstein condensates, one observes  the occurrence of quantized ``vortices" (which are local point defects of superconductivity or superfluidity, surrounded by a current loop). The vortices repel each other, while being confined together by the effect of the magnetic field or rotation,  and the result of the competition between  these two effects is that, as predicted by Abrikosov \cite{a}, they arrange themselves in a particular {\it triangular lattice} pattern, called {\it Abrikosov lattice}, cf. Fig. 
\ref{fig32} (for more pictures, see {\tt www.fys.uio.no/super/vortex/}).
\begin{figure}[ht!]
\begin{center}
\includegraphics[width=5cm]{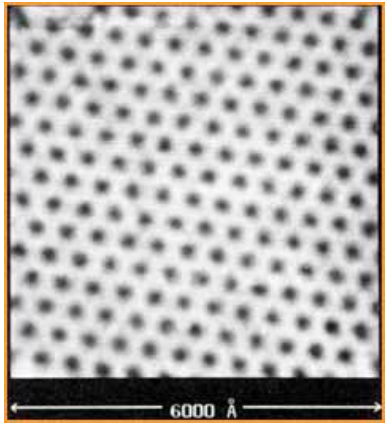}
\caption{Abrikosov lattice, H. F. Hess et al. Bell Labs
{\it Phys. Rev. Lett.} 62, 214 (1989)}
\label{fig32}

\end{center}

\end{figure}
When restricting to a two-dimensional situation, it can be shown formally, cf. \cite[Chap. 1]{ln} for a formal derivation, but also rigorously \cite{livre,ssgl}, that the minimization problem governing the behavior of such systems can be reduced, in terms of the vortices, to the minimization of an energy of the form \eqref{HN} in the case \LogD.
This naturally leads to the question of understanding the connection between minimizers  of \eqref{HN} + \eqref{wlog2d} and the Abrikosov triangular lattice.

  \section{Equilibrium measure and leading order behavior}
 
\subsection{The macroscopic behavior: empirical measure}
It is well-known since \cite{choquet} (see e.g. \cite{safftotik} for the logarithmic cases),  that under suitable assumptions on $V$, the leading order behavior of $\HN$ is governed by the minimization of the functional
\begin{equation} \label{definitionI}
\I_V(\mu) := \iint_{\R^\d \times \R^\d} \g(x-y) d\mu(x)d\mu(y) + \int_{\R^\d} V(x) d\mu(x)
\end{equation}
defined over the space $\mc{P}(\R^\d)$ of probability measures on $\R^\d$, which may also take the value $+ \infty$.

Note that $\I_V(\mu)$ is simply a continuum  version of the discrete Hamiltonian $\HN$. From the point of view of statistical mechanics, $\I_V$ is the ``mean-field" limit energy of  $\HN$, while we will see that from the point of view of probability, $\I_V$ plays the role of a {\it rate function}.

Under suitable assumptions, there is a unique minimizer of $\I_V$ on the space $\probas(\R^\d)$ of probability measures on $\R^\d$, it is  called the \textit{equilibrium measure} and we denote it by $\meseq$.  Its uniqueness follows from the strict convexity of $\I_V$.  For existence, one assumes
\begin{description}
\item[(A1)] $V$ is finite l.s.c. and bounded below
\item[(A2)] (growth assumption)
$$
   \underset{|x|\to + \infty}{\lim}\( \frac{V(x)}{2} + \g(x)\) = + \infty.
$$
\end{description}
One then proceeds in a standard fashion, taking a minimizing sequence for $\I_V$ and using that $\I_V$ is coercive and lower semi-continuous.
Finally, one has 
\begin{theo}[Frostman \cite{frostman}, existence  and characterization of the equilibrium measure] \label{theoFrostman} Under the assumptions \textbf{(A1)}-\textbf{(A2)}, the minimum of $\I_V$ over $\mc{P}(\R^\d)$ exists, is finite and is achieved by a unique $\meseq$, which  has a compact support of  positive capacity. In addition $\meseq$ is uniquely characterized by the fact that 
\begin{equation}
\label{EulerLagrange}
\left\lbrace
\begin{array}{cc} h^{\meseq} +\displaystyle \frac{V}{2} \geq c & \mbox{q.e. in } \R^\d  \vspace{3mm} \\ 
 h^{\meseq} + \displaystyle\frac{V}{2}= c & \mbox{q.e. in the support of }\meseq \end{array} \right.
\end{equation}
where \be \label{defhmu0}
h^{\meseq}(x) := \int_{\R^\d} \g(x - y) d\meseq(y)
\ee is the electrostatic potential generated by $\meseq$;
 and then the constant $c$ must be 
 \begin{equation}
\label{defc1} c = \I_V(\meseq) - \hal \int_{\R^\d} V(x) d\meseq(x).
 \end{equation}
 \end{theo}
The proof of this theorem can easily be adapted from \cite[Chap. 1]{safftotik} or \cite[Chap. 2]{ln}. The relations \eqref{EulerLagrange} can be seen as the Euler-Lagrange equations associated to the minimization of $\I_V$, they are obtained by making variations of the form $(1-t)\meseq+ t\nu$ where $\meseq$ is the minimizer, $\nu$ is an arbitrary probability measure of finite energy, and $t\in [0,1]$.

Here, the capacity of a set (see \cite{safftotik,ah} or \cite[Sec. 11.15]{liebloss}) is an appropriate notion of size, suffice it to say that $q.e.$ means except on a set of capacity zero, and that a set of null capacity has zero Lebesgue measure (but the converse is not true).

\begin{remark} \label{rem8} Note that by (\ref{coulombkernel}), in dimension $\d\ge 2$,  the function $h^{\meseq}$ solves 
\begin{displaymath}
- \Delta h^{\meseq} =\cd \meseq
\end{displaymath} where $\cd$ is the constant defined in \eqref{defcd}.
\end{remark}

\begin{example}[$C^{1,1}$ potentials and RMT examples]  In general,  the relations \eqref{EulerLagrange} say that the total potential $h^{\meseq}+ \frac{V}{2}$ is constant on the support of the charges. Moreover, in dimension $\d \ge 2$,  applying the Laplacian on both sides of \eqref{EulerLagrange}  and using  Remark~\ref{rem8} gives that, on the interior of the support of the equilibrium measure, if $V\in C^{1,1}$, \be \label{densmu0}
\cd \meseq = \frac{\Delta V}{2} 
\ee
 i.e. the density of the measure on the interior of its support is given by $\frac{\Delta V}{2\cd}$.  For example if $V$ is quadratic, then the density $ \frac{\Delta V}{2\cd} $ is constant  on the interior of its support.  This corresponds to the important examples of the Hamiltonians which arise in random matrix theory, more precisely~:

\begin{itemize}
\item in dimension $\d=2$, for $V(x) = |x|^2$, one may check that $\meseq = \frac{1}{\pi} \mathbf{1}_{B_1}$  where $\mathbf{1}$ denotes a characteristic function and $B_1$ is the unit ball, i.e.  the equilibrium measure is the normalized Lebesgue measure on the unit disk (by uniqueness, $\meseq$ should be radially symmetric, and the combination of \eqref{densmu0} with the constraint of being a probability measure imposes the support to be $B_1$). This is known as the {\it circle  law} for the Ginibre ensemble in the context of Random Matrix Theory (RMT). Its derivation is attributed to   Ginibre \cite{ginibre}, Mehta \cite{mehta}, an unpublished paper of Silverstein and  Girko \cite{girko1}.
\item in dimension $\d\ge 3$, the same  holds, i.e.  for $V(x)=|x|^2$  we have $\meseq = \frac{\d}{\cd} \indic_{B_{(\d-2)^{1/\d}   }}$ by the same reasoning.
\item in dimension $\d=1$, with $\g =-\log |\cdot |$ and $V(x) = x^2$, the equilibrium measure is $\meseq (x)= \frac{1}{2\pi} \sqrt{4-x^2} \mathbf{1}_{|x|\leq 2}$, which corresponds in the context of RMT (GUE and GOE ensembles) to {\it   Wigner's semi-circle law}, cf. \cite{wigner,mehta}.
\end{itemize}
\end{example}

The equilibrium measure $\meseq$ can also be interpreted in terms of the solution to a classical {\it obstacle problem}, which is essentially dual to the minimization of $\I_V$, and better studied from the PDE point of view (in particular the regularity of $\meseq $ and of the boundary of its support). For this aspect, see \cite[Chap. 2]{ln}.

\begin{defi}\label{def23} From now on, we denote by $\zeta$ the function
\begin{equation} \label{defzeta}
\zeta = h^{\meseq} + \frac{V}{2} - c.
\end{equation}
\end{defi}
\noindent We note that in view of \eqref{EulerLagrange}, $\zeta\ge 0$ a.e. and $\zeta=0$ $\meseq$-a.e.

In the rest of the course, we will always assume that the support of $\meseq$, denoted $\Sigma$, is a set with a nice boundary (say $C^{1, \alpha}$), and that the density $\meseq(x)$ of $\meseq$ is a H\"older continuous function in $\Sigma$.

\subsection{Large Deviations Principle at leading order} 
Let us start by recalling the basic definitions associated to Large Deviations Principles (cf. for instance \cite{dz}).

\begin{defi}[Rate function] \label{ratefun} Let $X$ be a metric space (or  a topological space). A rate function is a l.s.c. function $I : X \rightarrow [0, + \infty]$, it is called a {\it good rate function} if its sub-level sets $\{x, I(x) \leq \alpha\}$ are compact. 
\end{defi}

\begin{defi}[Large deviations]\label{definiLDP} Let $\{P_N\}_N$ be a sequence of Borel probability measures on $X$ and $\{a_N\}_N$ a sequence of positive real numbers diverging to $+ \infty$. Let also $I$ be a (good) rate function on $X$. The sequence $\{P_N\}_N$ is said to satisfy a large deviation principle (LDP) at speed $a_N$ with (good) rate function $I$ if for every Borel set $E \subset X$ the following inequalities hold : 
\be \label{343}
- \inf_{\overset{\circ}{E}}I \leq \underset{N \to + \infty}{\liminf} \frac{1}{a_N} \log P_N(E) \leq  \underset{N \to + \infty}{\limsup} \frac{1}{a_N} \log P_N(E) \leq - \inf_{\bar{E}} I
\ee
where $\overset{\circ}{E}$ (resp. $\bar{E}$) denotes the interior (resp. the closure) of $E$ for the topology of $X$.
\end{defi}
Formally, it means that $P_N(E)$  should behave roughly like $e^{-a_N \inf_{E} I}$. The rate function $I$ is the rate of exponential decay of the probability of rare events, and the events with larger probability are the ones on which $I$ is smaller.

For us, a convenient macroscopic observable is given by the empirical measure of the particles: if $\XN\in(\R^\d)^N$ we form
\begin{equation} \label{def:empmeasure}
\emp[\XN] := \frac{1}{N} \sum_{i=1}^N \delta_{x_i},
\end{equation}
which is a probability measure on $\R^\d$. The minimisation of $\I_V$ determines the macroscopic (or global) behavior of the system in the following sense:
\begin{itemize}
\item Minimizers of $\HN$ are such that $\emp[\XN]$ converges to $\meseq$ as $N \to \infty$ (cf. \cite[Chap. 2]{ln})
\item In fact $\emp[\XN]$ converges weakly to $\meseq$ as $N \ti$ \textit{almost surely under the canonical Gibbs measure} $\PNbeta$, see below.
\end{itemize}
In other words, not only the minimisers of the energy, but almost every (under the Gibbs measure) sequence of particles is such that the empirical measure converges to the equilibrium measure. Since $\meseq$ does not depend on the temperature, in particular \textit{the asymptotic macroscopic behavior of the system is independent of $\beta$.}
This is formally stated as the following Large Deviations Principle,  due to \cite{hiaipetz} (in dimension 2),  \cite{bg} (in dimension $1$) and  \cite{bz} (in dimension 2) for the particular case of a quadratic potential (and $\beta = 2$), see also \cite{berman} for results in a more general (still determinantal) setting of multidimensional complex manifolds, or  \cite{cgz} which recently treated more general singular $g$'s and $V$'s.    We present here the result  for the Coulomb gas in any dimension and general potential, which is not more difficult, for the proof see \cite{ln}. 

We need an additional assumption:

\textbf{(A3)} Given $\beta$, for $N$ large enough, we have
\begin{align*}  
(\LogU, \ \LogD) &\int \exp\left(-\beta N \left( \frac{V(x)}{2}- \log |x|\right)\right) \, dx <\infty  \\ 
(\Coul) & \int \exp\left(-\frac{\beta}{2}  N  V(x)\right) \, dx < + \infty.
\end{align*}

Note in particular that  \textbf{(A3)} ensures that the integral in \eqref{defZ} is convergent, hence
 $Z_{N, \beta}$ well-defined.

\begin{theo}[Large deviations principle for the Coulomb gas at speed $N^2$] \label{LDP}
\mbox{}
Let $\beta>0$ be given and assume that $V$ is continuous,  satisfies \textbf{(A3)} and  that $(1-\alpha_0)V$ satisfies  \textbf{(A2)} for some $\alpha_0>0$. 
Then the sequence $\{\mathbb{P}_{N, \beta}\}_N$ of probability measures on $\mc{P}(\R^\d)$ satisfies a large deviations principle at speed $N^2$ with good rate function $\frac{\beta}{2}\hat{\I}_V$ where $\hat{\I}_V = \I_V - \min_{\mc{P}(\R^\d)}  \I_V= \I_V - \I_V(\meseq)$. Moreover 
\begin{equation}
\lim_{N\to + \infty} \frac{1}{N^2} \log Z_{N, \beta} = - \frac{\beta}{2} \I_V(\meseq) = - \frac{\beta}{2} \min_{\mc{P}(\R^\d)} \I_V.
\end{equation}
\end{theo} Here the underlying topology is  that of weak convergence on $\mc{P}(\R^\d)$.

The heuristic reading of the LDP is that 
\be  \label{heurisldp}
\P_{N, \beta}(E) \approx e^{-\frac{\beta}{2} N^2 (\min_{E} \I_V - \min \I_V)}, \ee 
which in view of the uniqueness of the minimizer of $\I_V$ implies as stated above that configurations whose empirical measure does not converge to $\meseq$ as $N\to \infty$ have exponentially decaying probability.

\subsection{Further questions}

In contrast to the macroscopic result, several observations (e.g. by numerical simulation, see Figure \ref{fig:fig}) suggest that the behavior of the system at microscopic scale\footnote{Since the $N$ particles are typically confined in a set of order $O(1)$, the microscopic, inter-particle scale is $O(N^{-1/d})$.} depends heavily on $\beta$. 
\begin{figure}[h!] \label{fig:fig}
\begin{minipage}[c]{.46\linewidth}
\begin{center}
\includegraphics[scale=0.12]{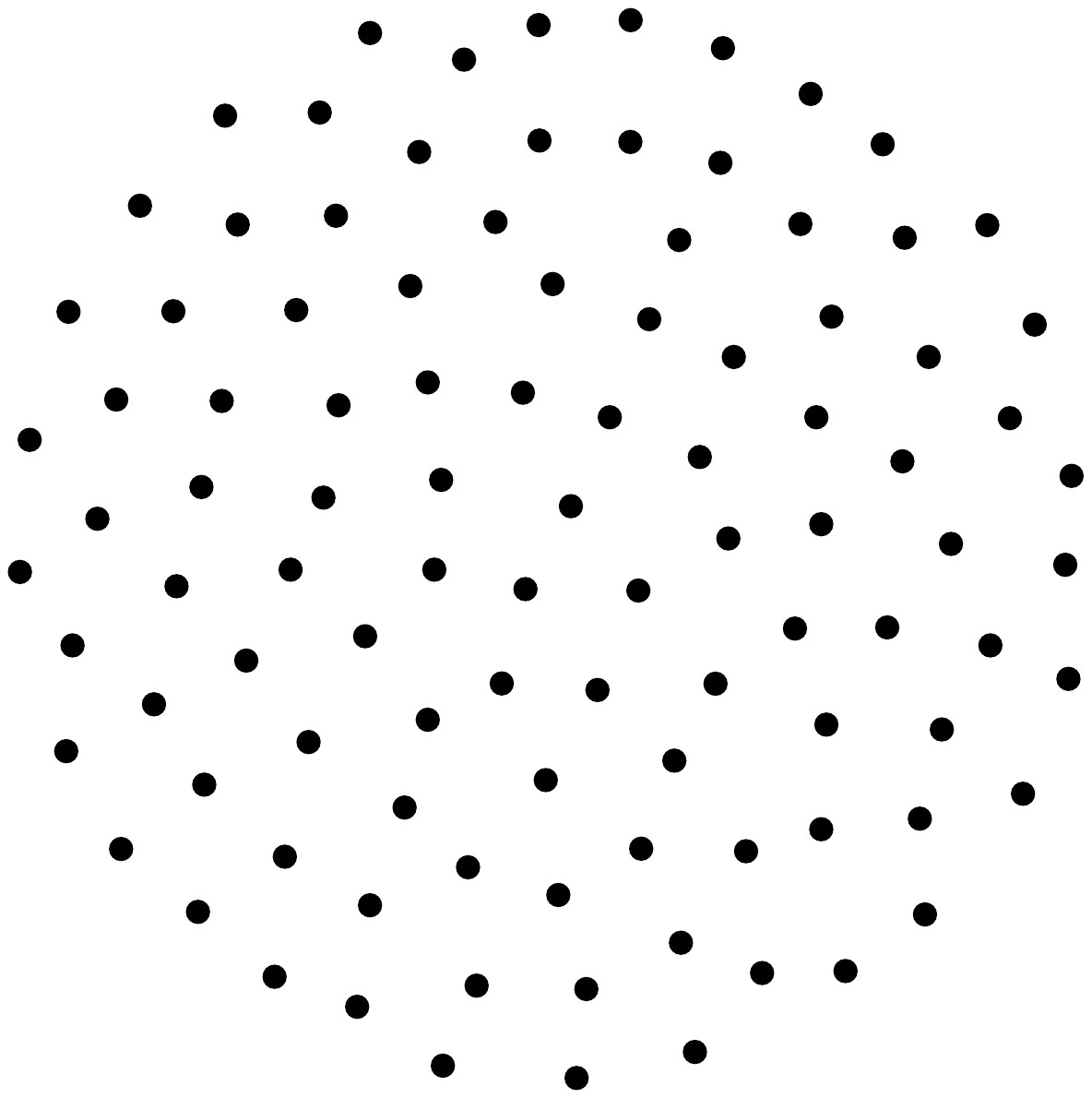} 
\end{center}
\end{minipage}
\begin{minipage}[c]{.46\linewidth}
\begin{center}
\includegraphics[scale=0.13]{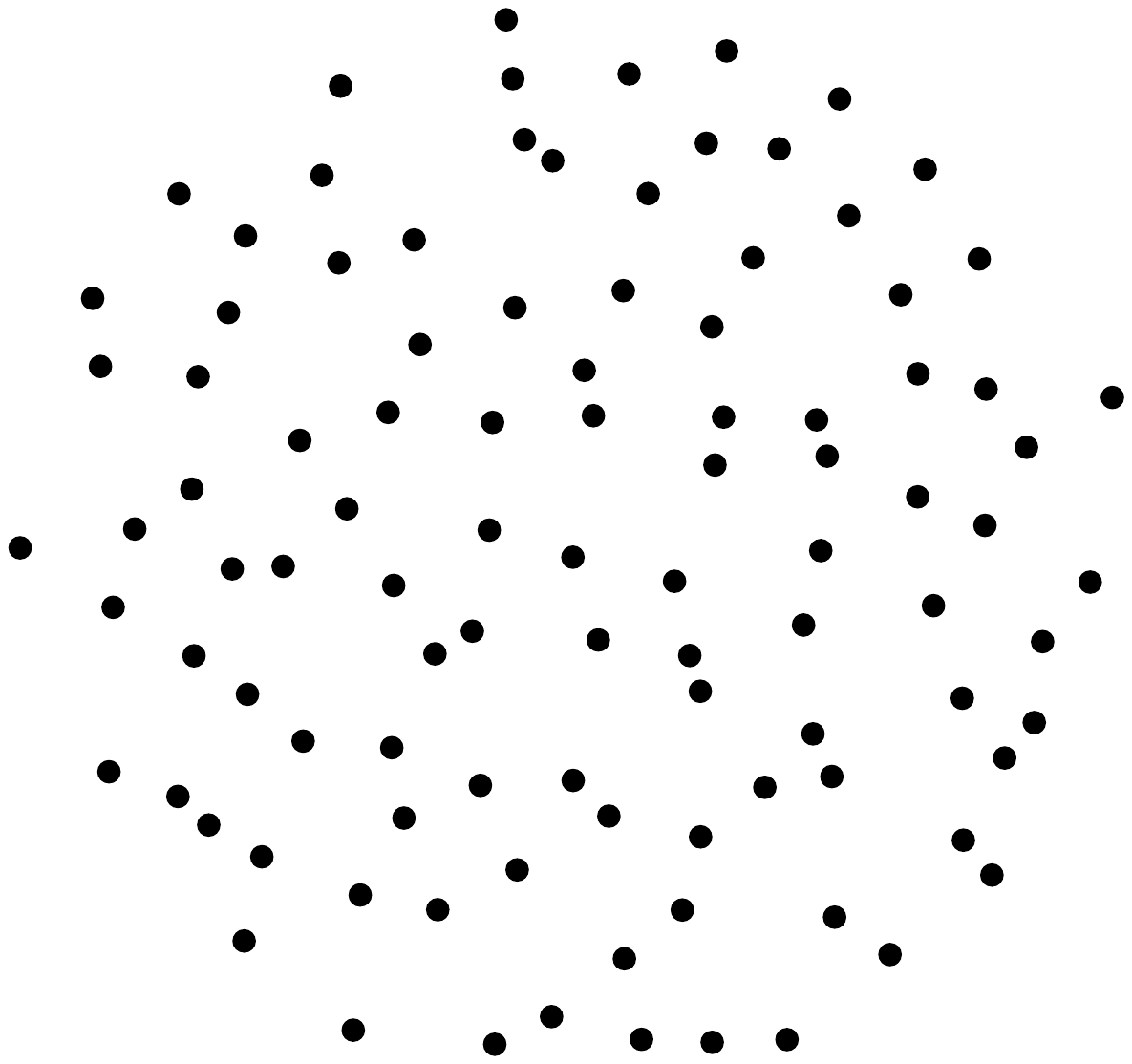} 
\end{center}
\end{minipage}
\vspace{-0.5cm}
\caption{Case \LogD \ with $N = 100$ and $V(x) = |x|^2$, for $\beta = 400$ (left) and $\beta = 5$ (right).}
\end{figure}

The questions that one would like to answer are then to describe the system beyond the macroscopic scale, at the mesoscopic (i.e. $N^{-\alpha}$ for $0<\alpha<1/\d$) scale, and at the microscopic ($N^{-1/\d}$) scale.
\\
Since one already knows that $\sum_{i=1}^N\delta_{x_i}- N\meseq$ is small (more precisely $\ll N$ in some appropriate norm), one knows that 
the so-called {\it discrepancy} $$D(x,r):= \int_{B_r(x) } \sum_{i=1}^N \delta_{x_i} - N  \, d\meseq$$ is  $o(N)$ as long as $r>0$ is fixed. 
Is this still true at the mesoscopic scale for $r$ of the order $N^{-\alpha}$ with $\alpha<1/\d$? Is it true down to the microscopic scale, i.e. for $r=RN^{-1/\d}$ with $R \gg 1$? Does it hold regardless of the temperature?  This would correspond to a {\it rigidity result}.
Note that point processes with discrepancies growing like the perimeter of the ball have been called {\it hyperuniform} and are of interest to physicists for a variety of applications (cf. \cite{To}). 

Once one proves rigidity down to the microscopic scale, one may also like to characterize the fluctuations of linear statistics of the form
$$\sum_{i=1}^N f(x_i)- N\int fd\meseq,$$ where $f$ is a regular enough test-function.
In the logarithmic cases, they are proven to converge to a Gaussian distribution whose variance depends on the temperature, as will be seen below. 

Another point of view is that of large deviations theory. Then, one wishes to study a microscopic observable, the microscopic point process obtained after blow-up by $N^{1/\d}$, and characterize it as minimizing a certain energy (or rate function, in the case  with temperature), thus connecting to a crystallization question. 
 
In all the cases, one wants to understand precisely how the behavior depends on $\beta$, but also on $V$. It is believed that at the macroscopic and microscopic levels, in the logarithmic cases  the behavior is independent on $V$, a  universality feature.

\section{Splitting of the Hamiltonian and  electric approach} \label{sectionaref} 
We now start to present the approach to these problems initiated in \cite{ss1} and continued in \cite{ss2,rs,ps,lebles}. It relies on a splitting of the energy into a fixed leading order term and a next order term expressed in terms of the charge fluctuations, and on a rewriting of this next order term  via  the ``electric potential" generated by the points.

\subsection{The splitting formula}

The splitting consists in an exact formula that separates the leading ($N^2$) order term  in $\HN$ from next order terms. 

Since we expect $\emp[\XN]$ to converge to $\meseq$, 
we may try to ``expand" around $\meseq$.
In all the sequel,  we denote for any probability measure $\mu$, 
\begin{equation} \label{deffluct}
\fluct_N^{\mu} [\XN]=N( \emp[\XN]  -  \mu)   = \sum_{i=1}^N \delta_{x_i} - N\mu.
\end{equation}
Unless ambiguous, we will drop the $\XN$ dependence.

\begin{lem}[Splitting formula]Assume $\meseq $ is absolutely continuous with respect to the Lebesgue measure. 
For any $N$ and any $\XN \in (\R^\d)^N$ we have
\begin{equation}\label{split0}
 \HN(\XN) =  N^2 \mathcal{I}_{V}(\meseq) + 2N \sum_{i=1}^N \zeta(x_i) +F_N^{\meseq} (\XN)
\end{equation}
where $\triangle$ denotes the diagonal of $\R^\d\times \R^\d$ and we define for any probability measure $\mu$
\be\label{def:FN} 
F_N^{\mu} (\XN)= \iint_{\triangle^c} \g(x-y)\, d\fluct_N^{\meseq} (x) d\fluct_N^{\meseq}(y).\ee
\end{lem}
\begin{proof}
We may write \begin{eqnarray}
\nonumber \HN(\XN)  & = &  \sum_{i \neq j} \g(x_i- x_j) + N \sum_{i=1}^N V(x_i)\\
\nonumber & = & N^2 \iint_{\triangle^c} \g(x-y) d\emp(x)  d\emp(y) + N^2 \int_{\R^\d} V d\emp(x) \\
\nonumber 
& = &  N^2 \iint_{\triangle^c} \g(x-y) d\meseq(x) d\meseq(y) + N^2 \int_{\R^\d} V d\meseq
\\ 
\nonumber
 & + &  2N \iint_{\triangle^c} \g(x-y) d\meseq(x) d\fluct_N(y)
+ N \int_{\R^\d} V d\fluct_N \\ 
\label{finh}
& + & \iint_{\triangle^c} \g(x-y) d\fluct_N(x)d\fluct_N(y).
\end{eqnarray}
We now recall that $\zeta$ was defined in \eqref{defzeta} by
\be
\zeta = h^{\meseq} + \frac{V}{2} - c = \int_{\R^\d} \g(\cdot-y)\, d\meseq(y)  + \frac{V}{2} - c\ee
and that $\zeta=0$  in $\Sigma$ (with the assumptions we made, one can check that  $\zeta$ is continuous, so the q.e. relation can be upgraded to everywhere).

With the help of this we may rewrite the medium line in the right-hand side of \eqref{finh} as  
\begin{multline*}
2N \iint_{\triangle^c} \g(x-y) d\meseq(x) d\fluct_N(y) + N \int_{\R^\d} V d\fluct_N \\
 = 2N  \int_{\R^\d} (h^{\meseq} + \frac{V}{2}) d\fluct_N = 2N  \int_{\R^\d} (\zeta + c) d\fluct_N \\
 = 2N^2 \int_{\R^\d} \zeta d\emp - 2N^2 \int_{\R^\d} \zeta d\meseq+ 2 N c \int_{\R^\d}  d\fluct_N = 2N^2 \int_{\R^\d} \zeta d\emp.
\end{multline*}
The last equality is due to the facts that  $\zeta \equiv 0$ on the support of $\meseq$ and that  $\emp$ and $ \meseq$ are both probability measures.   We also have to notice that since  $\meseq$ is absolutely continuous with respect to the Lebesgue measure, we may include the diagonal  back into the domain of integration.
By that same argument, one may recognize in the first line of the right-hand side of \eqref{finh}, the quantity $N^2 \I_V(\meseq)$, cf. \eqref{definitionI}. 
\end{proof}
The function $\zeta$ can be seen as an effective potential, whose sole role is to confine the points to the set $\Sigma$. We have thus reduced to studying $F_N^{\meseq}$. It is a priori not clear of which order this term is, and whether it is bounded below!

\subsection{Electric interpretation}
To go further, we use an electric interpretation of the energy $ F_N^{\meseq}$, as first used in \cite{ss1}, and  the rewriting of the energy via truncation as in \cite{rs,ps}. 

Such a computation allows to replace the sum of pairwise interactions of all the charges and ``background" by an integral (extensive) quantity, which is easier to handle.  This will be the first time that the Coulomb nature of the interaction is really used. 





\subsubsection{Electric potential and truncation}
\label{sec:potential}
{\it Electric potential.}
For any $N$-tuple $\XN$ of points in the space $\R^\d$, and any  probability density $\mu$, we define the (electric) potential generated by $\XN$ and $\mu$ as
\begin{equation}\label{def:HNmu}
H_N^{\mu} (x) :=\int_{\R^\d} - \g (x-y)\, \left(\sum_{i=1}^N \delta_{x_i} - N d\mu\right)(y).
\end{equation} Note that in principle it should be denoted $H_N^{\mu}[\XN](x)$, but we omit the $\XN$ dependence for the sake of lightness of notation.
The potential $H_N^{\mu}$ satisfies 
\begin{equation}\label{eq:HNmu}
-\Delta H_{N}^{\mu}=\cd \left(\sum_{i=1}^N \delta_{x_i} - N d\mu\right) \quad \text{in} \ \R^\d.\end{equation}
Note that $H_N^{\mu}$ decays at infinity, because the charge distribution $\fluct_N$ is compactly supported and has zero total charge, hence, when seen from infinity behaves like a dipole. More precisely, $H_N^\mu$ decays like $\nabla \g$ at infinity, that is $O(\frac{1}{|x|^{\d-1}})$ and its gradient $\nabla H_N^\mu$ decays like the second derivative $D^2\g$, that is $O(\frac{1}{|x|^{\d}})$.


{\it Truncated potential.}
Let $\XN \in (\R^\d)^N$ be fixed. For any $\vec{\eta} = (\eta_1, \dots, \eta_N)$ 
we define the truncated potential
\begin{equation}\label{def:HNmutrun}
H_{N,\vec{\eta}}^{\mu} (x)= H_N^{\mu}(x)-\sum_{i=1}^N \(\g(x-x_i)-\g(\eta_i) \)_+
\end{equation}
where $(\cdot)_+$ denotes the positive part of a number. Note that in view of the singular behavior of $\g$ at the origin, $H_N^{\mu}$ diverges at each $x_i$, and here we ``chop off" these infinite peaks at distance $\eta_i$ from $x_i$.
We will also denote
\be\label{def:truncation} \f_{\eta} (x)= (\g(x)-\g(\eta))_+,\ee and point out that $\f_\eta$ is supported in $B(0,\eta)$.
We note that
\begin{equation}
\label{eqhne} H_{N, \vec{\eta}}^{\mu} (x) := \int_{\R^2} \g(x-y) \left(\sum_{i=1}^N \delta_{x_i}^{(\eta_i)}  -N d\mu\right)(y),
\end{equation}
where $\delta_{x_i}^{(\eta_i)}$ denotes the uniform measure of mass $1$ on $\pa B(x_i, \eta_i)$.
By $H_{N,\eta}$ we simply denote $H_{N,\vec{\eta}}$ when all the $\eta_i$ are chosen equal to $\eta$.




\subsubsection{Re-expressing the interaction term}
Formally, using Green's formula (or Stokes' theorem) and the definitions, one would like to write that  in the Coulomb cases\begin{equation} \label{formalcomputation}
F_N^{\mu}(\XN) = \int H_N^{\mu} d\fluct_N= \int H_N^{\mu} (- \frac{1}{\cd} \Delta H_N^{\mu}) \approx \frac{1}{\cd} \int |\nabla H_N^{\mu}|^2 \end{equation}This is the place where we really use for the first time in a crucial manner the Coulombic  nature of the interaction kernel $\g$.
Such a computation allows to replace the sum of pairwise interactions of all the charges and ``background" by an integral (extensive) quantity, which is easier to handle in some sense. However, \eqref{formalcomputation} does not make sense because  $\nabla H_N^\mu$ fails to be in $L^2$ due to the presence of Dirac masses.  Indeed, near each atom $x_i$, the vector-field $\nabla H_N^\mu$ behaves like $\nabla \g$ and the integrals $\int_{B(0,\eta)} |\nabla \g|^2$  are divergent in all dimensions.   Another way to see this is that the Dirac masses charge the diagonal $\triangle$ and so $\triangle^c$ cannot be reduced to the full space. The point of the truncation above is precisely to remedy this and give a way of computing $\int |\nab H_N^{\mu}|^2$ in a ``renormalized" fashion. 

We now give a proper meaning to the statement.
\begin{lem} \label{lem:monoto} Given $\XN$, for any $\vec{\eta}$ with $\eta_i \le \frac12$ such that  the $B(x_i,\eta_i)$ are disjoint, for any absolutely continuous probability measure $\mu$  we have
 \begin{equation}\label{eqlm}
 F_N^{\mu}(\XN) =\frac{1}{\cd} \left(\int_{\R^\d}|\nab H_{N, \vec{\eta}} ^{\mu}|^2  -\cd \sum_{i=1}^N \g(\eta_i)\) + O \( N \|\mu\|_{L^\infty}  \sum_{i=1}^N \eta_i^2\).  
\end{equation}
If the balls $B(x_i, \eta_i)$ are not assumed to be disjoint, then we still have that the left-hand side is larger than the right-hand side.
\end{lem}
The last term in the right-hand side of \eqref{eqlm} should be thought of as a small error. In practice, we will take $\eta_i \le N^{-1/\d} \eta$ and let $\eta \to 0$. The error is then $O(\eta^2 N^{2-2/\d})$.

\begin{proof} For the proof, we drop the superscripts $\mu$. 
First we notice that $\int_{\R^\d}|\nab H_{N, \vec{\eta}}|^2 $ is a convergent integral and that 
\be\label{intdoub}\int_{\R^\d}|\nab H_{N, \vec{\eta}}|^2=\cd \iint \g(x-y)d\( \sum_{i=1}^N \delta_{x_i}^{(\eta_i)} -N\mu\)(x) d\( \sum_{i=1}^N \delta_{x_i}^{(\eta_i)} -N\mu\)(y).\ee
Indeed, we may choose $R$ large enough so that all the points of $\XN$ are contained in the ball $B_R = B(0, R)$.
    By Green's formula  and \eqref{eqhne}, we have
\be\label{greensplit1}
\int_{B_R} |\nabla H_{N,\vec{\eta}}|^2 \\= \int_{\partial B_R} H_{N,\vec{\eta}} \frac{\partial H_N}{\partial \nu} - \cd\int_{B_R} H_{N,\vec{\eta}} \(   \sum_{i=1}^N \delta_{x_i}^{(\eta_i)}-N\mu\)  .
\ee
In view of the decay of $H_N$ and $\nab H_N$ mentioned above, the boundary integral tends to $0$ as $R \to \infty$, and so we may write 
$$\int_{\R^\d}|\nab H_{N, \vec{\eta}} |^2= \cd \int_{\R^\d} H_{N,\vec{\eta}} \(   \sum_{i=1}^N \delta_{x_i}^{(\eta_i)}-N\mu\)  $$ and thus \eqref{intdoub} holds.
We may next write 
\begin{multline}
\label{l3}\iint \g(x-y)d\( \sum_{i=1}^N \delta_{x_i}^{(\eta_i)} -N\mu\)(x) d\( \sum_{i=1}^N \delta_{x_i}^{(\eta_i)} -N\mu\)(y)\\
-\iint_{\triangle^c} \g(x-y)\, d\fluct_N(x)\, d\fluct_N(y)\\
 = \sum_{i=1}^N \g(\eta_i) + \sum_{i\neq j} \iint \g(x-y) \(\delta_{x_i}^{(\eta_i)} \delta_{x_j}^{(\eta_j)} -   \delta_{x_i}\delta_{x_j}\)\\+2 N\sum_{i=1}^N\iint \g(x-y)\( \delta_{x_i}-\delta_{x_i}^{(\eta_i)} \) \mu.\end{multline}
Let us now observe that $\int \g(x-y)\delta_{x_i}^{(\eta_i)} (y)$, the potential generated by $\delta_{x_i}^{(\eta_i)}$ is equal to $\int \g(x-y) \delta_{x_i}$ outside of $B(x_i,\eta_i)$, and smaller otherwise. Since its  Laplacian is $-\cd \delta_{x_i}^{(\eta_i)}$, a negative measure,  this  is also a superharmonic function, so by the maximum principle, its value at a point $x_j$ is larger or equal to its average on a sphere centered at $x_j$. Moreover, outside $B(x_i,\eta_i)$ it is a harmonic function, so its values are equal to its averages. We deduce from these considerations, and reversing the roles of $i $ and $j$,  that for each $i\neq j$,
$$\int\g(x-y) \delta_{x_i}^{(\eta_i)} \delta_{x_j}^{(\eta_j)} \le \int\g(x-y) \delta_{x_i} \delta_{x_j}^{(\eta_j)}
\le \int\g(x-y) \delta_{x_i} \delta_{x_j}$$
with equality if $B(x_i,\eta_i) \cap B(x_j,\eta_j) = \emptyset.$
We conclude that the second term  in the right-hand side of \eqref{l3} is nonnegative and equal to  $0$ if all the balls are disjoint. Finally, by the above considerations,  since $\int \g(x-y) \delta_{x_i}^{(\eta_i)}$ coincides with $\int \g(x-y)\delta_{x_i}$ outside $B(x_i,\eta_i)$, 
we may rewrite the last term in the right-hand side of \eqref{l3} as 
$$2 N\sum_{i=1}^N\int_{B(x_i,\eta_i)} ( \g(x-x_i)- \g(\eta_i)) d\mu.$$ 
On the other hand, recalling \eqref{def:truncation}, if $\mu\in L^\infty$ then  we have
\be\label{fdmu}
\left|\sum_{i=1}^N \int_{\R^\d}\f_{\eta_i} d\mu\right|\le  C_\d \|\mu\|_{L^\infty} \sum_{i=1}^N\eta_i^2.\ee
Indeed, it suffices to observe that 
\be\label{intfeta}\int_{B(0,\eta)} \f_{\eta}= \int_0^\eta (\g(r)-\g(\eta))r^{\d-1}\, dr= -\int_0^r \g'(r) r^\d\, dr ,\ee
with an integration by parts (and an abuse of notation, i.e. viewing $\g$ as a function on $\R$) and using  the explicit form of $\g$ it follows that 
\be\label{bfeeta}\int_{B(0,\eta)} |\f_{\eta}|\le C_\d \eta^2.\ee
By the definition \eqref{def:truncation}, we thus have obtained the result. 
\end{proof}

\subsection{The case $\d=1$}
In the case $\g(x)=-\log |x|$ in dimension $1$ \eqref{wlog}, or in the Riesz cases \eqref{kernel2},  
 $\g$ is no longer the Coulomb kernel, so the formal computation \eqref{formalcomputation} does not work. 
However $\g$ is in the case \eqref{wlog} the kernel of the half-Laplacian, and it is known that the half-Laplacian can be made to correspond to the Laplacian by adding one extra space dimension. In the same way, in the case \eqref{kernel2}, $\g$ is the kernel of a second order local operator, after adding one extra space dimension.
In other words, in the case \eqref{wlog} we should imbed the  space $\R$ into the two-dimensional space $\R^2$ and consider the harmonic extension of $H_N^{\meseq}$, defined in \eqref{def:HNmu}, to the whole plane. That extension will solve an appropriate Laplace equation, and we will reduce dimension $1$ to a special case of dimension $2$. An analogue procedure, due to Caffarelli-Silvestre \cite{caffsilvestre}
 applies to the case \eqref{kernel2}.
 This is the approach that was proposed in \cite{ss2} for the \LogU \ case  and in \cite{ps} for the \Riesz\  case.
 
 Let us now get more specific about the extension procedure in the case \LogU. We view $\R$ as identified with $\R\times \{0\} \subset \R^{2}= \{ (x,y), x\in \R, y\in \R\}$. 
 Let us denote by $\delta_{\R}$ the uniform measure on $\R\times \{0\}$, i.e. such that for any smooth $\varphi(x,y) $ (with $x\in \R, y \in \R$) 
we have 
$$\int_{\R^2} \varphi \delta_{\R}= \int_{\R  } \varphi(x,0) \, dx.$$

 Let us still consider $\mu$ a measure on $\R$ (such as $\meseq$ the equilibrium measure on $\R$ associated to $\I_V$ 
as in Theorem \ref{theoFrostman}).

Given $x_1, \dots, x_N \in \R$, as explained above we identify them with  the  points $(x_1, 0), \dots, (x_N,0)$ in $\R^{2}$, and
we may then define the potentials $H_N^{\mu}$ and truncated potentials $H_{N,\vec{\eta}}^{\mu}$  in $\R^{2}$ by 
\bm
H_N^\mu = \g * \left(\sum_{i=1}^N \delta_{(x_i,0)} - N\mu \delta_{\R}\right)\qquad H_{N,\vec{\eta}}^{\mu} = \g * \left(\sum_{i=1}^N \delta_{(x_i,0)}^{(\eta)} - N \mu \delta_{\R}\right).
\em 
Since $\g$ is naturally extended to a function in $\R^{2}$, these potentials make sense as functions in $\R^{2}$ and 
$H_N^{\mu} $ solves 
\be \label{bbe} -\Delta  H_N^{\mu}=2\pi \left(\sum_{i=1}^N \delta_{(x_i,0)} - N\mu \delta_{\R}\right).
\ee
$H_N^{\mu}$  is nothing else than the harmonic extension to $\R^2$, away from the real axis,  of the potential  defined in dimension $1$ by the analogue of \eqref{def:HNmu}. This is closely related to  the {\it Stieltjes transform}, a commonly used object in Random Matrix Theory (more precisely the gradient of $H_N^{\mu}$ is like the Stieltjes transform).

The proof of Lemma \ref{lem:monoto} then goes through without change, if one replaces 
$\int_{\R^\d} |\nab H_{N,\vec{\eta}}^\mu|^2$ with $\int_{\R^{2}} |\nab H_{N,\vec{\eta}}^\mu|^2 .$

\subsection{The electric energy controls the fluctuations}
Since \eqref{eq:HNmu} (resp. \eqref{bbe}) holds, we can immediately relate the fluctuations to $H_N^{\mu}$, and via the  Cauchy-Schwarz  inequality, control them by the electric energy.
\begin{prop}\label{prop:fluctenergy}
Let  $\varphi$ be a compactly supported Lipschitz function in $\R^\d$ supported in $U$, and  $\mu$ be a bounded probability density on $\R^\d$. Let  $\vec{\eta}$ be a $N$-tuple of distances such that $\eta_i \le N^{-1/\d}$, for each $i = 1, \dots, N$. For each configuration $\XN$, we have
\begin{equation} \label{controlfluctuations}
\left|\int_{\R^\d} \varphi \, \fluct_N^{\mu}  \right|
  \le C \|\nab\varphi\|_{L^\infty}\( |U|^{\hal}  \|\nab H_{N,\vec{\eta}}^{\mu} \|_{L^2(U)}+  N^{1-\frac{1}{\d}} \right)
\end{equation}
where $C$ depends only on $\d$.

\end{prop}
\begin{proof}In the 1D logarithmic case, we first
extend
 $\varphi$ to a smooth compactly supported test function in $\R^2$  coinciding with $\varphi (x) $ for any $(x,y)$ such that $|y|\le 1$ and equal to $0$ for $|y|\ge 2$. 
 

In view of \eqref{eq:HNmu} (resp. \eqref{bbe}) and applying Cauchy-Schwarz, we  have 
\begin{equation}\label{rel3}
\left|\int  \Big( \sum_{i=1}^N \delta_{x_i}^{(\eta_i)} - N\mueq\Big) \varphi\right|= \frac{1}{\cd}\left| \int_{\R^\d}\nabla H_{N,\vec{\eta}}^{\mueq} \cdot \nab \varphi\right|\le C |U|^{\hal} \|\nab \varphi\|_{L^\infty} \|\nab H_{N,\vec{\eta}}^{\mueq}\|_{L^2(U)},
\end{equation}(resp. with an integral over $U\times \R$ in the 1D log case). 
Moreover, since $\eta_i \le N^{-1/\d}$, we have
\begin{equation}
\left|\int \left(\fluct_N^\mu- \Big( \sum_{i=1}^N \delta_{x_i}^{(\eta_i)} - N\mu\Big) \right) \varphi\right|= 
\left|\int  \Big( \sum_{i=1}^N( \delta_{x_i} - \delta_{x_i}^{(\eta_i)}) \Big) \varphi\right|\le N^{1-\frac{1}{\d}} \|\nab \varphi\|_{L^\infty} .\end{equation}
The result follows.
\end{proof}

\begin{coro}
Applying this with $\eta_i=N^{-1/\d}$, we deduce in view of  Lemma \ref{lem:monoto}  that 
\begin{multline} \label{controlfluctuations2}
\left|\int_{\R^\d} \varphi \, \fluct_N^{\mu}  \right|
 \\  \le
  \begin{cases} C\|\nab\varphi\|_{L^\infty} \( |\mathrm{Supp} \, \varphi|^{\hal}   \( F_N^\mu(\XN)+ \frac{N}{\d} \log N + CN\|\mu\|_{L^\infty}   \)^{\hal}+ N^{1-\frac{1}{\d}}\)
   &  \text{for \LogU, \LogD} \\
    C \|\nab\varphi\|_{L^\infty}\(|\mathrm{Supp} \, \varphi|^{\hal}   \(F_N^{\mu} (\XN)+ N^{2-\frac{2}{\d}}\|\mu\|_{L^\infty} \right)^{\hal}   + N^{1-\frac{1}{\d}}\) & \text{for \Coul},
\end{cases} \end{multline} where $C$ depends only on $\d$ and $\|\mu\|_{L^\infty}$.
\end{coro}

\subsection{Consequences for the  energy and partition function}
Thanks to the splitting, we can expand the energy to next order. For instance, 
combining \eqref{split0} and Lemma \ref{lem:monoto}, choosing again  $\eta_i= N^{-1/\d}$,  we easily obtain 
\begin{coro}[First energy lower bound]
If $\meseq \in L^\infty$, then for any $\XN$
\be\label{flb}
\HN(\XN)\ge
 N^2 \I_V(\meseq) +2N\sum_{i=1}^N\zeta(x_i) -\(  \frac{N}{\d}\log N\) \indic_{\LogU, \LogD} - C\|\meseq\|_{L^\infty} N^{2-2/\d} ,
 \ee
where $C$ depends only on $\d$. \end{coro}
 
Let us show how this lower bound easily  translates  into an upper bound for the partition function $Z_{N, \beta}$   in the case with temperature. In the cases \Coul\,  it is better to normalize the energy differently and define 
\begin{equation}\label{pnb2} \PNbeta= \frac{1}{\ZNbeta} e^{- \frac{\beta}{2} N^{\min(\frac{2}{\d}-1,0) } \HN(\XN)} d\XN\end{equation}
and define $\ZNbeta$ accordingly. Note that in the case \LogU, \LogD, this does not change anything.

\begin{coro}[An easy upper bound for  the partition function]\label{boundlogz}
Assume that $V$ is continuous, such that $\meseq$ exists, satisfies \textbf{(A3)} and has an $L^\infty$ density. Then for all $\beta >0$, and for $N$ large enough, we have
$$   \log Z_{N,\beta}  \leq     - \frac{\beta}{2} N^{\min (2, \frac{2}{\d}+1)  }   \I_V(\meseq) + \Big(\frac{\beta}{4} N \log N\Big) \indic_{\LogU,\LogD}  + C(1+\beta)N$$
where $C$ depends only on $\meseq$ and the dimension.
\end{coro}

To prove this, let us  state a lemma that we will use repeatedly and that exploits assumption \textbf{(A3)}.

\begin{lem} \label{asymptozeta} Assume that $V$ is continuous, such that $\meseq$ exists, and satisfies \textbf{(A3)}. We have 
\be
\lim_{N\to  + \infty} \left(\int_{(\R^\d)^N} e^{- \beta N \sum_{i=1}^N \zeta(x_i)} d\XN\right)^{\frac1N} = |\omega|\ee where $\omega=\{\zeta=0\}$. \end{lem}
\begin{proof} First, by separation of variables, we have
$$ \left(\int_{(\R^\d)^N}  e^{- \beta N \sum_{i=1}^N \zeta(x_i)} d\XN\right)^{\frac1N}= \int_{\R^\d}
e^{- \beta N \zeta(x)} \, dx.$$
Second, we recall that  since $\meseq$ is a compactly supported probability measure, $h^{\meseq}$ must asymptotically behave like $\g(x)$ as $|x|\to \infty$, thus $\zeta= h^{\meseq} + \frac{V}{2}-c $  grows like $\g(x)+ \hal V -c$.  The assumption \textbf{(A3)} thus ensures that  for  $N$ large enough, $\int e^{- \beta N\zeta(x)}\, dx<+\infty.$ 
Moreover, by definition of $\omega$,    
$$e^{- \beta N \zeta}\to \indic_{\omega} \quad \text{as} \ N\to+\infty$$
  pointwise and monotonically, and $\omega$ has finite measure in view of the growth of $h^{\meseq}$ and thus of  $\zeta$. The monotone convergence theorem  allows to conclude.
\end{proof}

\begin{proof}[Proof of the corollary]
Inserting \eqref{flb} into the definition of $\ZNbeta$ we are led to 
\begin{multline}
 \log Z_{N,\beta} \leq    N^{ \min ( \frac{2}{\d}-1,0)}\(- \frac{\beta}{2} N^2 \I_V(\meseq) + \Big(\frac{\beta}{2\d} N \log N\Big)\indic_{\LogU, \LogD} + C\beta N^{2-\frac2\d} \)\\+ \log \(\int e^{-N\beta \sum_{i=1}^N \zeta(x_i)} d\XN\).\end{multline} Using Lemma \ref{asymptozeta} to handle the last term, we deduce that 
\begin{multline*}
 \log Z_{N, \beta} \leq   N^{ \min (\frac{2}{\d}-1,0)}\( - \frac{\beta}{2} N^2 \I_N(\meseq) + 
\Big(\frac{\beta}{2\d} N \log N\Big)\indic_{\LogU,\LogD} + C\beta N^{2-\frac2\d}\) \\ + N(\log |\omega|+o_N(1))
\end{multline*}
which gives the conclusion.
\end{proof}

The converse inequality to Corollary \ref{boundlogz} actually holds and we have 
\begin{prop}\label{dlogzn}
Assume that $V$ is continuous, such that $\meseq$ exists,  satisfies \textbf{(A3)} and has an $L^\infty$ density. Then for all $\beta >0$, and for $N$ large enough, we have
\begin{equation}\label{dlogz}
\log Z_{N,\beta}  =    - \frac{\beta}{2}     N^{ \min (2, \frac{2}{\d}+1)}\I_V(\meseq) + \Big(\frac{\beta}{2\d} N \log N\Big) \indic_{\LogU, \LogD}  +  O(N).\end{equation}
where $O(N)$ depends only on $\beta,\meseq$ and the dimension.
\end{prop}
One of the final goals on this course will be to present a more precise result with 
the constant  in  the  order  $N$  term  identified and characterized variationally.  This will also provide information on the behavior of the configurations at the microscale.
Anticipating the result, we will be able to show that in the logarithmic cases, $\log \ZNbeta$ has an expansion where the $V$ dependence can be decoupled: 
\begin{multline} \label{logz0}
 \log \ZNbeta =- \frac{\beta}{2} N^2 \mathcal \I_V(\meseq) +\frac{\beta}{2}\frac{N \log N}{\d} 
 -N C(\beta,  \d)
 \\- N \left( 1 -\frac{\beta}{2 \d}\right) \int_\Sigma \meseq(x)\log \meseq(x) \, dx + o((\beta+1)N),
 \end{multline} 
where $C(\beta, \d)$ is a constant depending only on $\beta$ and  the dimension (1 or 2).
Obtaining \eqref{dlogz} alone without an explicit constant is easier, and we will assume it for now. 

In the sequel, it is convenient to use the splitting to rewrite 
\begin{equation}\label{PNrewri}
d\PNbeta(\XN)= \frac{1}{\KNbeta(\meseq,\zeta)} e^{-\frac{\beta }{2} N^{\min(\frac{2}{\d}-1,0) }\( F_N^{\meseq}(\XN)+2N\sum_{i=1}^N \zeta(x_i) \) }d\XN\end{equation}
where for any function $\xi$ growing sufficiently fast at infinity, 
$\KNbeta(\mu, \xi)$ is defined as 
\begin{equation}\label{def:KNbeta}\KNbeta(\mu,\xi)= \int e^{-\frac{\beta }{2} N^{\min(\frac{2}{\d}-1,0) }\( F_N^{\mu}(\XN)+2N\sum_{i=1}^N \xi(x_i) \) }d\XN.\end{equation}
In view of what precedes, we have 
 \begin{equation}
 \label{devkn}
 \KNbeta(\meseq, \zeta)= \ZNbeta e^{\frac{\beta}{2}N^{\min(\frac{2}{\d}+1,2) }\I_V(\meseq)}.\end{equation}

\subsection{Consequence: concentration bounds}
With Proposition \ref{dlogzn}, we deduce an upper bound on the exponential moments of the electric energy $F_N$:
\begin{coro}
We have, for some constant $C$ depending on $\beta$ and $V$,
\begin{equation}\label{claimeq}
\left| \log \Esp_{\PNbeta} \left[ \exp\left(\frac{\beta}{4} \left(N^{\min(\frac{2}{\d}-1,0) }  F_N^{\meseq}(\XN)   +(\frac{N}{\d}\log N)\indic_{\LogU, \LogD} \right) \right)\right] \right|\le C N.
\end{equation}
\end{coro}
\begin{proof}
We may write   
\begin{multline*}
\Esp_{\PNbeta} \left[ \exp\left(\frac{\beta}{4} N^{\min(\frac{2}{\d}-1,0) } F_N^{\meseq}(\XN) \right)\right] \\ = \frac{1}{\KNbeta(\meseq,\zeta)} \int \exp \left( -\frac{\beta}{4} \left(N^{\min(\frac{2}{\d}-1,0) }  ( F_N^{\meseq}(\XN)  + 2N \sum_{i=1}^N 2 \zeta(x_i))\right) \right) d\XN 
\\
=  \frac{K_{N, \frac{\beta}{2}}(\meseq, 2\zeta)}{\KNbeta(\meseq, \zeta)}.
\end{multline*} 
Taking the $\log$ and using \eqref{devkn}, \eqref{dlogz} to expand both terms up to order $N$ yields the result.
\end{proof} 
As proved in \cite{loiloc} in the case \LogD  (and as should be true in higher dimensions as well) this control can also be improved  into a local  control at all {\it mesoscales}: by this we mean the control 
\begin{equation}\label{claimeqloc}
\left| \log \Esp_{\PNbeta} \left[ \exp\left(\frac{\beta}{4}\int_{U} |\nab H_{N,\vec{\eta}}|^2\) \right] \right|\le C|U| N,
\end{equation} for suitable $\vec{\eta}$ and cubes $U$ of sidelength $N^{-\alpha}$, $\alpha<\hal$.

Combining \eqref{claimeq} with \eqref{controlfluctuations2}, we may then easily obtain concentration bounds for the fluctuations:
\begin{coro}[Concentration of fluctuations]\label{coro311}
For any Lipschitz function $\varphi$, we have
\begin{equation}
\log \Esp_{\PNbeta}\( e^{ t\(\int  \varphi \fluct_N^{\meseq}  \)^2 }\) \le CN\end{equation}
where the constant $C$ depends on $t $, $\beta$ and $\varphi$.\end{coro}
This already gives us a good control on the fluctuations: by Markov's inequality we may for instance deduce that the probability that $\left|\int \fluct_N^{\meseq} \varphi\right| \gg \sqrt N$ is exponentially small (the control on the exponential moment is a stronger information though), or other moderate deviations  bounds  (compare with the  total mass of $\fluct_N^{\meseq}$ which is order $N$).  This also improves on the a-priori concentration bound in $N\log N$ generally obtained in the one-dimensional logarithmic case by the non-electric approach \cite{BorGui1,mms}. We  also refer to  \cite{chm} for a recent improved result in the same spirit valid in all Coulomb cases.

This is however not the best control one can obtain. The rest of these notes will be devoted to two things: obtaining a better control on the fluctuations when $\varphi$ is assumed to be  more regular, and obtaining a full LDP at next order for empirical fields, providing the exact  constant in the order $N$ term of  \eqref{dlogz}.

\section{CLT for fluctuations in the logarithmic cases}
In this section, we restrict our attention to the logarithmic cases \LogU, \LogD.
\\
Let us define the fluctuations of the linear statistics associated to $\xi_N$ as the random variable
\begin{equation} \label{def:FluctN}
\Fluct_N(\xi) := \int_{\R^2} \xi\, d\fluct_N^{\meseq},
 \end{equation}
 where $\fluct_N^{\meseq}$ is as in \eqref{deffluct}.
 The goal of this section is to present a result that $\Fluct_N(\xi)$ converges in law to a Gaussian random variable with explicit mean and variance.  This is achieved by proving that the Laplace transform of $\Fluct_N(\xi)$ converges to that of the Gaussian. 
 
 This approach was pioneered in  \cite{joha} and immediately
 leads to dealing with a Coulomb gas with  potential $V $ replaced by $V_t=V+t \xi$ with $t$ small. We then  follow here the approach of  \cite{ls2,bls} where we  use  a simple change of variables which is a transport map between the equilibrium measure $\mueq$ and  the equilibrium measure $\mueqt$ for the perturbed potential.  Note that the use of changes of variables in this context is not new, cf. \cite{joha,BorGui1,shch,bfg}.  In our approach, it essentially replaces the use of the so-called ``loop" or Dyson-Schwinger equations.

\subsection{Reexpressing the fluctuations as a ratio of partition functions}
We start by introducing the notation related to the perturbed potential and equilibrium measure.
\begin{defi} For any $t \in \R$, we define 
\begin{itemize}
 \item The perturbed potential $V_t$ as $V_t :=V+t \xi$ and the 
perturbed equilibrium measure $\mueqt$. \item The next-order confinement term $\zetat :=  \zeta_{V_t}$, as in \eqref{defzeta}.
\item The next-order energy $F_N^{\mueqt}(\XN)$ as in \eqref{def:FN}.
\item The next-order partition function $\KNbeta(\mueqt, \zetat)$ as in \eqref{def:KNbeta}.
\end{itemize}
\end{defi}

\begin{lem}[Reexpressing the Laplace transform of fluctuations]
For any $t \in \R$ we have
\begin{multline}\label{Laplace}
\Esp_{\PNbeta}\left[ \exp\left( -\frac{\beta}{2} Nt  \Fluct_N(\xi) \right) \right] \\ = \frac{\KNbeta(\mueqt, \zetat)}{\KNbeta(\meseq, \zeta)} \exp\left( - \frac{\beta}{2} N^2 \left(\mathcal{I}_{\Vt} (\mueqt) - \mathcal{I}_V(\meseq) - t \int \xi d\meseq  \right)\right).
\end{multline}
\end{lem}

\begin{proof}  
First, we notice that, for any $t$ in $\R$
   \begin{equation} \label{lt1a}
   \Esp_{\PNbeta} \left[ \exp(-\frac{\beta}{2}Nt\Fluct_N(\xi) )\right]= \frac{Z_{N,\beta}^{\Vt}}{\ZNbeta} \exp\left(\frac{\beta}{2} N^2t  \int \xi \, d\meseq\right)
   \end{equation}
Using the splitting formula \eqref{split0} and the definition of $\KNbeta$ as in \eqref{def:KNbeta} we see that for any $t$
\begin{equation} \label{Pnbeta2a}
\KNbeta(\mueqt, \zetat) = \ZNbeta^{\Vt} \exp\left(\frac{\beta}{2} N^2 \mathcal{I}_{\Vt}(\mueqt)\right), 
\end{equation}
thus combining  \eqref{lt1a} and \eqref{Pnbeta2a}, we obtain \eqref{Laplace}.
\end{proof}
To compute the limit of the Laplace transform of $\Fluct_N(\xi)$ we will just need to apply the formula \eqref{Laplace} with $t=- \frac{2s}{\beta N}$ for some arbitrary number $s$, hence $t$ will indeed be very small. 

We will see below that the term in the exponential in \eqref{Laplace} is computable in terms of $\xi$ and then there will  remain to study the ratio of partition functions $ \frac{\KNbeta(\mueqt, \zetat)}{\KNbeta(\meseq, \zeta)}$.

\subsection{Transport and change of variables}

Our method introduced in \cite{ls2,bls} is based on the construction of a transport $\phi_t$ such that $\phi_t\# \mueq= \mueqt$, with $\phi_t= \id + t \psi$. In fact  the transport chosen will not satisfy exactly $\phi_t\# \mueq= \mueqt$ but will satisfy it approximately as $t\to 0$, and  in these notes will neglect the error for simplicity.

In order to construct the correct transport map, one first needs to understand the perturbed equilibrium measure $\mueqt$, or in other words, to understand how $\mu_V$ varies when $V$ varies. 
In the case $\LogU$, the dependence of $\mueq$ on $V$ is indirect, but well understood. At least in the ``one-cut case" (i.e. when the support of $\mueq$  is a single interval), the right transport map (as $t \to 0$) is  found  by inverting with respect to $\xi$ a so-called ``master operator" $\Xi$, which already arose in the Dyson-Schwinger approach \cite{BorGui1,BorGui2,shch,bfg}.

 In dimension 2 (and in fact  in all Coulomb cases), the perturbed equilibrium measure is easy to compute when $\xi$ is supported in $\Sigma$: one can check that  it is $\mueqt=\meseq-\frac{t}{2\pi \beta} \Delta \xi $, and a correct transport map is $\id + t \psi$ with $\psi= -\frac{\nab \xi}{\cd\mueq}$.  The case where $\xi$ is not supported in $\Sigma$ i.e. has a support intersecting $\partial \Sigma$ is much more delicate: one needs to understand how $\partial \Sigma$ is displaced under the perturbation. This is described precisely (in all dimensions) in \cite{serser}. The PDE approach used there replaces  Sakai's theory used in \cite{ahm}, which is restricted to the two-dimensional analytic setting.
 
In the interior (two-dimensional) Coulomb case, it follows from a direct computation based on the exact formula for $\mu_{V_t}$  that 
\begin{equation}
\label{versv}
\lim_{N\to \infty, t= -\frac{2s}{\beta N}}  
 - \frac{\beta}{2} N^2 \left(\mathcal{I}_{\Vt} (\mueqt) - \mathcal{I}_V(\meseq) - t \int \xi d\meseq  \right)
 = \frac{s^2}{2\pi \beta} \int_{\R^2} |\nab \xi|^2.\end{equation} In the general case, thanks to the analysis of \cite{serser}, we find that $\int |\nab \xi|^2 $ gets replaced by $\int_{\R^2}|\nab \xi^\Sigma|^2$, where $\xi^\Sigma$ denotes the harmonic extension of $\xi$ outside $\Sigma$.
 
In the one-dimensional case, a more delicate computation based on the above facts reveals that 
\begin{equation}\label{versv1}
\lim_{N\to \infty, t= -\frac{2s}{\beta N}}  
 - \frac{\beta}{2} N^2 \left(\mathcal{I}_{\Vt} (\mueqt) - \mathcal{I}_V(\meseq) - t \int \xi d\meseq  \right)
= \frac{s^2}{2}\int \xi'\psi d \meseq,\end{equation}
where $\psi$ is the transport map defined above. We have thus identified the  limit of the exponential term in \eqref{Laplace} and now turn to the ratio of partition functions.

Once the (approximate) transport has been constructed and proven to be regular enough (it is typically once less differentiable than $\xi$) and  such that we approximately have $\zeta_t \circ \phi_t=\zeta$, 
using   the change of variables $y_i=\phi_t(x_i)$ we find
\begin{multline}\label{Kcv}
\KNbeta(\mueqt,\zetat) = \int  \exp \Big( - \frac{\beta}{2}\Big(F_N^{\mueqt}(\Phi_t(\XN)) + 2N \sum_{i=1}^N\zetat\circ \phi_t(x_i)   \Big) \\+ \sum_{i=1}^N \log |\det D\phi_t|(x_i) \Big) \ d \XN   \\ 
 = \int \exp \Big( - \frac{\beta}{2}   \Big(F_N^{{\mueqt}}(\Phi_t(\XN)) + 2N \sum_{i=1}^N\zeta(x_i)   \Big) + \sum_{i=1}^N \log |\det D\phi_t|(x_i)  \Big) \ d \XN   .
\end{multline}

Thus 
\begin{equation}\label{rpf} \frac{\KNbeta( \mueqt,\zetat)}{\KNbeta (\meseq,\zeta)}=\Esp_{\PNbeta} \(e^{-\frac{\beta}{2} (F_N^{{\mueqt}}(\Phi_t(\XN))- F_N^{\mueq} (\XN)) +\sum_{i=1}^N \log |\det D\Phi_t|(x_i)}\)\end{equation}
hence we need to evaluate the difference of  energies $F_N^{{\mueqt}}(\Phi_t(\XN))-F_N^{\meseq} (\XN)$.

\subsection{Energy comparison}
Let us present that  computation in the \LogU\  case for simplicity, knowing that it has a more complex analogue in the \LogD\  case.

\begin{lem}\label{44}
 For any probability density $\mu$, any $\XN\in \mathbb{R}^N$, and any  $\psi\in C^2_c(\R)$, defining
 \begin{equation}\label{ani}
 \Ani[\XN,\psi,\mu]=\iint\frac{\psi(x)-\psi(y)}{x-y} \, d\fluct_N^{\mu} (x) \, d \fluct_N^{\mu}(y)
 \end{equation} and letting
 $\Phi_t(\XN) = (\phi_t(x_1),\cdots, \phi_t(x_N))$,
  we have
 \begin{multline} \label{tmuttotmueqz}
\left| F_N^{\phi_t\# \mu}(\Phi_t(\XN))-F_N^{\mu}(\XN) - \sum_{i=1}^N \log\phi_t' (x_i) + \frac{t}{2} \, \Ani[\XN,\psi,\mu]\right|
\\ \le  C t^2 \left( F_N^{\mu} (\XN)+N\log N\right).
 \end{multline}
 \end{lem}
 This is the point that essentially replaces the loop equations.
 \begin{proof}
 We may write
  \begin{multline*}
 F_N^{\phi_t\# \mu}(\Phi_t(\vec{X}_N))-F_N^{\mu}(\vec{X}_N) \\ = 
- \int_{\mathbb{R}^2 \backslash \triangle} \log |x-y|\Big(\sum_{i=1}^N \delta_{\phi_t (x_i)}- N \phi_t\# \mu\Big)(x) \Big(\sum_{i=1}^N \delta_{\phi_t (x_i)}- N \phi_t \#\mu\Big)(y)   \\
 +   \int_{\mathbb{R}^2 \backslash \triangle} \log |x-y| d\fluct_N^{\mu}(x) d\fluct_N^{\mu}(y) \\ 
 = -  \int_{\mathbb{R}^2 \backslash \triangle} \log \frac{|\phi_t(x) - \phi_t(y)|}{|x-y|} d\fluct_N^{\mu}(x) d\fluct_N^{\mu}(y) \\
 = -  \int_{\mathbb{R}^2 } \log \frac{|\phi_t(x) - \phi_t(y)|}{|x-y|} d\fluct_N^{\mu}(x) d\fluct_N^{\mu}(y) + \sum_{i=1}^N \log |\phi_t'(x_i)|.
\end{multline*}
Using that by definition $\phi_t = \id + t \psi$ where $\psi$ is in $C^2_c(\mathbb{R})$, and  the fact that 
\begin{equation*}
\left\|\frac{\log(1+tx) - tx}{t^2}\right\|_{C^2(K)} \leq C_K
\end{equation*}
for all compact sets $K$, some constant $C_K$ and $t$ small enough, we get by the chain rule 
\begin{equation*}
\log \frac{|\phi_t(x) - \phi_t(y)|}{|x-y|} = t \ \frac{\psi(x) - \psi(y)}{x-y} + t^2 \  \varepsilon_t(x,y) ,
\end{equation*}
with $\|\varepsilon_t\|_{C^2(\R^2)}$  uniformly bounded in $t$. Applying Proposition \ref{prop:fluctenergy} twice, we get that
$$
\left| \iint \varepsilon_t(x,y) d\fluct_N^{\mu}(x) d\fluct_N^{\mu}(y) \right| \leq C t^2 \left( F_N^{\mu} (\XN)+N\log N\right).
$$
This yields the result in the \LogU \ case.
 \end{proof}

\subsection{Computing the ratio of partition functions}

Inserting the result of  Lemma~\ref{44} (or its two-dimensional analogue) into  \eqref{rpf} we obtain 
\begin{multline}\label{411}
\frac{\KNbeta( \mueqt,\zetat)}{\KNbeta (\meseq,\zeta)}\\
= \Esp_{\PNbeta}
\Big( \exp\Big( (1-\frac{\beta}{2\d}) \sum_{i=1}^N \log |\det D \phi_t|(x_i) + \frac{\beta  t}{4}\Ani[\XN,\psi,\meseq] \\+ O(t^2 (F_N^{\meseq} (\XN)+ N\log N) ) \Big) \Big)\\ = 
\Esp_{\PNbeta} \Big( \exp\Big( (1-\frac{\beta}{2\d}) \Fluct_N [\log |\det D \phi_t|] 
+ \frac{\beta  t}{4}\Ani[\XN,\psi,\meseq] \\
+ O(t^2 (F_N^{\meseq} (\XN)+ N\log N) ) \Big) \Big)\\ \times  e^{ (1-\frac{\beta}{2\d}) N\int \log |\det D \phi_t|  d\meseq} .\end{multline}
Let us now examine the terms in this right-hand side. 

First, remembering that $\phi_t\# \meseq= \mueqt$ (approximately), we have $\det D \phi_t= \frac{\meseq}{\mueqt\circ \phi_t}$ and thus, by definition of the push-forward,
\begin{equation}\label{epf}
\int \log |\det D\phi_t | d\meseq= \int \log \meseq d \meseq- \int \log \mueqt d\mueqt,\end{equation} which is of order $t $ i.e. $O(\frac{s}{N})$, as $N\to \infty$. More precisely, we may compute that  in the \LogU \ case
\begin{equation}\label{mean2}
\lim_{N\to \infty, t= -\frac{2s}{\beta N}  }\(1-\frac{\beta}{2}\)  N\(\int \meseq \log \meseq- \int \mueqt\log \mueqt\)= -  \(1-\frac{\beta}{2}\) \frac{2s}{\beta}\int \psi' d\meseq\end{equation}
and in  the \LogD\  interior case the analogue is $\frac{s}{2\pi} \(\frac{1}{\beta}- \frac{1}{4}\)\int \Delta \xi\log \Delta V$.

There now remains to evaluate all the terms in the expectation.  
We note that by the Cauchy-Schwarz inequality, we have for instance
\begin{equation}\label{abc}\log \E(e^{a+b+c}) \le \log \E(e^{4a})+\log \E(e^{4b})+\log \E(e^{4c}).\end{equation}
First,  by \eqref{claimeq} and H\"older's inequality we have the control 
\begin{equation}\label{epf2}\log \Esp_{\PNbeta}\(\exp(O(t^2 (F_N^{\meseq} (\XN)+ N\log N) ) )\)= O( \frac{s^2}{N^2}  N) =  O(\frac{s^2}{N}) \end{equation}   (up to changing  $\beta$ in the formula), with  the choice $t=O(s/N)$.

Since $\phi_t= \id+ t\psi$ with $\psi$ regular enough, we find by combining  \eqref{controlfluctuations2} and \eqref{claimeq} that 
$$\log \Esp_{\PNbeta}\( \Fluct_N [\log |\det D \phi_t|] \)\le O(t \sqrt{N}) = O(\frac{s}{\sqrt{N}}).$$

Next, we turn to the   term   $\Ani$ (which we call anisotropy), which  is the most delicate one.
There are two ways to handle this term and conclude. The first is direct and works in the one-dimensional one-cut regular case, i.e. when the operator $\Xi$  that finds the right transport map is always invertible. This is described in the appendix of \cite{bls} and ressembles the method of \cite{BorGui1}.
The second way works in all logarithmic (one-dimensional possibly multi-cut or critical, and two-dimensional) and calls instead the information on $\log \ZNbeta$ obtained in \eqref{logz0}. It is described in \cite{bls} in the one-dimensional case and \cite{ls2} in the two-dimensional case.
Let us now present each.
\subsection{Conclusion in the one-dimensional one-cut regular case}
We note that the term $A$ can be seen as  being essentially also a fluctuation term with 
\begin{equation}\label{anir1}
\Ani[\XN, \psi,\meseq]= \int f(x) d\fluct_N(x)  \qquad f(x):= \(\int \hat \psi(x,y)d \fluct_N (y)\) \end{equation}
where $$\hat\psi(x,y):= \frac{\psi(x)-\psi(y)}{x-y}$$ is as regular as the regularity of $\xi$ allows (note that this fact is only true in dimension 1 and so the argument breaks in dimension 2!). 
Using Proposition \ref{prop:fluctenergy} twice, we obtain 
\begin{equation}\label{nabg}
\|\nab f\|_{L^\infty}\le   \left|\int \nab_x \hat \psi(x,y)d \fluct_N (y)\right|\le C\|\nab_x\nab_y \hat \psi\|_{L^\infty} \(F_N^{\meseq}(\XN)+ N \log N+ C N\)^{\hal}\end{equation} and 
\begin{multline*}
|\Ani [\XN,\psi,\meseq]|=\left|\int f(x) d\fluct_N(x)\right|\le C \|\nab f\|_{L^\infty}  \(F_N^{\meseq}(\XN)+ N \log N+ C N\)^{\hal}\\ \le
C \|\hat \psi\|_{C^2}   \(F_N^{\meseq}(\XN)+ N \log N+ C N\).\end{multline*}
In view of \eqref{claimeq}, we deduce that 
$$\left|\Esp_{\PNbeta} \left[-\frac{\beta}{2} t \Ani[\XN,\psi,\meseq]\right]\right|\le O(\frac{s}{N} N) = O(s).
$$
Combining all these elements and inserting them into \eqref{411},  we deduce that $\frac{\KNbeta( \mueqt,\zetat)}{\KNbeta (\meseq,\zeta)}$ is of order $1$ as $N \to \infty$, and inserting into \eqref{Laplace} together with \eqref{versv}-\eqref{versv1} we obtain 
that $\Esp_{\PNbeta}(e^{s\Fluct_N(\xi)})$ is bounded (in terms of $s$ and $\xi$) as $N\to \infty$. We may then bootstrap this information by applying it to $f$, first improving    \eqref{nabg}  into 
\begin{multline*}\|D^k f\|_{L^\infty} \le   \left|\int D^k_x \hat \psi(x,y) d\fluct_N (y)\right|\\ \le C\|D^k_x\nab_y \hat \psi\|_{L^\infty} \(F_N^{\meseq}(\XN)+ N \log N+ C N\)^{\hal} \end{multline*}
obtaining (with H\"older's inequality)  if $k$ is large enough, that 
$$\log \Esp_{\PNbeta} \( e^{ t\Fluct_N(f) }\)\le  C t\sqrt{N} = O(\frac{s}{\sqrt{N}}). $$
We then conclude that in fact 
\begin{equation}\label{concla}
\left|\log \Esp_{\PNbeta} \left[-\frac{\beta}{2} t \Ani[\XN,\psi,\meseq]\right]\right|\le o_N(1).\end{equation}

Combining all the previous relations and \eqref{abc}, \eqref{versv}--\eqref{versv1}, and inserting into   \eqref{Laplace}, we conclude that
$$\lim_{N\to \infty} \log \Esp_{\PNbeta} \( e^{s\Fluct_N(\xi)}\)= m_{\xi} s+ v_{\xi} s^2$$
where $m_{\xi}$ is the  mean and $v_{\xi}$ is a variance given by the explicit formulae above. 

 We have thus obtained  that the Laplace transform of $\Fluct_N(\xi)$ converges to that of a Gaussian of mean $m_\xi$ and variance $v_\xi$, concluding the proof in the one-dimensional one-cut regular case.
 Note that 
 reinserting into \eqref{411} we have in fact obtained a precise expansion of  $\frac{\KNbeta( \mueqt,\zetat)}{\KNbeta (\meseq,\zeta)}$, which also allows to directly evaluate 
 $\frac{d}{dt}_{t=0} \log \KNbeta( \mueqt,\zetat)$ up to order $o(N)$. But this is valid for all variations of the type $V+t\xi$ of the potential, hence by interpolating between two potentials $V_0$ and $V_1$ which are such that their equilibrium measures are both one-cut regular, we can compute 
 $\log \KNbeta(\mu_{V_1}, \zeta_1)-\log \KNbeta(\mu_{V_0}, \zeta_0)$ up to an error $o(N)$ and obtain effectively the formula \eqref{logz0}.
 Moreover, since all the terms arising in the exponent in \eqref{411} are essentially fluctuations, the reasoning can be bootstrapped by inserting the CLT result and obtaining an expansion to higher order of $\frac{d}{dt}_{t=0} \log \KNbeta( \mueqt,\zetat)$. This is possible because the operator $\Xi $ is always invertible in the one-cut case, allowing to build the transport map associated to the fluctuation of any function, provided that function is regular enough. This way, one may obtain a relative expansion of $\log \ZNbeta$ with an error of arbitrary order, provided the potentials are regular enough, which is  what is found in \cite{BorGui1}. 
 
 \subsection{Conclusion in the two-dimensional case or in the general one-cut case}
 In these cases, we control the term $A$ differently: instead we use the important relation \eqref{logz0} which we have assumed so far, and  which provides another way of computing the left-hand side of \eqref{411}. Comparing these two ways applied to a small, fixed $t$, we find on the one-hand 
 $$\log \frac{\KNbeta( \mueqt,\zetat)}{\KNbeta (\meseq,\zeta)}= \(1-\frac{\beta}{2\d}\)  N\(\int \meseq \log \meseq- \int \mueqt\log \mueqt\) + o(N)$$
 and on the other hand combining \eqref{411}, \eqref{epf}, \eqref{epf2},
 \begin{multline*}\log \frac{\KNbeta( \mueqt,\zetat)}{\KNbeta (\meseq,\zeta)}=  \(1-\frac{\beta}{2\d}\)  N\(\int \meseq \log \meseq- \int \mueqt\log \mueqt\) \\+ 
 \log \Esp_{\PNbeta} \( \exp \frac{\beta t}{4} A[\XN,\psi, \mueq]\) + o(N)\end{multline*} 
 This implies that 
 $$\log \Esp_{\PNbeta} \( \exp \frac{\beta t}{4} A[\XN,\psi, \mueq]\) = o(N),$$
 and this is true for any $t$ small enough. Applying then H\"older's inequality leads to 
 $$\log \Esp_{\PNbeta} \( \exp \frac{ -s }{2N} A[\XN,\psi, \mueq]\) = o(1).$$
 This replaces \eqref{concla} and allows to conclude. The variance is given by the right-hand side of \eqref{versv} and \eqref{versv1} respectively, the mean is given by the right-hand side of \eqref{mean2} or its two-dimensional analogue.

 In both cases, we have obtained the following 
 \begin{theo}
 Assume $\xi$ is regular enough. In the one-dimensional multi-cut or critical cases, assume in addition that $\xi$ is in the range of $\Xi$, and an analogous condition in the two-dimension case with several connected components.  
 Then $\Fluct_N(\xi)$ converges in law to a Gaussian random variable of explicit mean and variance.\end{theo}
 
 Let us make several additional comments:
 \begin{itemize}
 \item In dimension 1, this theorem was first proven in \cite{joha} for polynomial $V$ and $\xi$ analytic. It was later generalized in \cite{shch,BorGui1,BorGui2,bl,llw}, still with strong assumptions on $\xi$, and to the critical case in \cite{bls}.
 \item If the extra conditions do not hold, then the CLT is not expected to hold. Rather, the limit should be a Gaussian convoled with a discrete Gaussian variable, as shown in the \LogU \ case in \cite{BorGui2}.
 \item In dimension 2, 
 the precise form of the variance as $C\int |\nab \xi|^2$ (or more generally $C\int |\nab \xi^\Sigma|^2$), means that $\fluct_N$ converges to a so-called Gaussian Free Field.
This result was proven for the determinantal case $\beta=2$ in \cite{ridervirag} (for $V$ quadratic) and \cite{ahm} under analyticity assumptions. It was then proven for all $\beta$ simultaneously in \cite{ls2} and \cite{bbny2}, with $\xi$ roughly assumed to be $C^4$. Note that the result also applies in the case where $\xi$ is supported at a mesoscale, i.e. $\xi_N(x)= \xi(x N^\alpha)$ with $\alpha<\hal$ (the analogous result in dimension 1 is proven in \cite{beklod}). 
 The results of \cite{ahm,ls2} are the only ones that also apply to the case where $\xi$ is not supported inside $\Sigma$.
 \item This theorem is to be compared to Corollary \ref{coro311}. It shows that if $\xi$ is smooth enough, the fluctuations $\Fluct_N(\xi)$ are in fact much smaller than could be expected. They are typically of order $1$, to be compared with the sum of $N$ iid random variables which is typically or order $\sqrt{N}$. This a manifestation of rigidity, which even holds down to the mesoscales.
 \item Moderate deviation bounds similar to those of \cite{bbny1} can also easily be obtained as a by-product of the method presented above.
  \end{itemize}

\section{The renormalized energy}

The goal of this section is to define a ``renormalized energy" or jellium energy, that  will arise as a rate function for the next order Large Deviation Principle that will be presented below.

It allows to compute  a total Coulomb interaction for an infinite system of discrete point ``charges" in a constant neutralizing background of  fixed density $m>0$. 
Such a system is often called a (classical) {\it jellium} in physics. 
 
 Starting again from the equations defining the electric potentials, \eqref{eq:HNmu} and \eqref{bbe}, we first perform a blow-up at scale $N^{1/\d}$ (the typical lengthscale of the systems). Denoting $x_i'=N^{1/\d} x_i$ the blown-up points of the configuration, and the blown-up equilibrium measure  ${\meseq}' (x)= \meseq(xN^{-1/\d})$, we then define for shortcut  the ``electric fields" 
 \begin{equation}\label{ele} E_N= \nab {H_N^{\meseq}} ' \qquad \nab {H_N^{\meseq}}'= \g * \( \sum_{i=1}^N \delta_{x_i'} -{\meseq}' \drd\),\end{equation}
 where the measure $\drd$ is needed only in the one-dimensional case, cf \eqref{bbe}.  This way, $E_N$ solves the  linear equation
\begin{equation}
\label{defhneta1}
-\div E_N = \cd \Big( \sum_{i=1}^N \delta_{x'_i} -{ \meseq}'\drd\Big) ,
\end{equation}  in which it is easy, at least formally, to pass to the limit $N\to \infty$. Note that this is written when centering the blow-up at the origin. 
When $\meseq$ is continuous, the blown-up measure ${\meseq}'$ then converges to the constant $\meseq(0)$. If another origin $x$ was chosen, the constant would be $\meseq(x)$, the local density of the neutralizing background charge. 

On the other hand,  as $N\to +\infty$, the number of points  becomes infinite  and they fill up the whole space, their local density remaining bounded (by control on the fluctuations). This way $\sum_{i=1}^N\delta_{x_i}$ is expected to converge to a distribution 
$\C\in \config(\Rd)$, where 
for $A$ a Borel set of $\Rd$ we denote by $\config(A)$ the set of locally finite point configurations in $A$ or equivalently the set of non-negative, purely atomic Radon measures on $A$ giving an integer mass to singletons. 
We mostly use $\C$ for denoting a point configuration and we will write $\C$ for $\sum_{p \in \C} \delta_p$.

We thus wish to define the energy associated to an electric field $E$ solving an equation of the form 
\be \label{eqclam}
-\div E= \cd\(\C- m\drd\)\quad \text{in} \ \R^{\d+k}\ee
where $\C\in \config(\R^\d)$ and $m$ is a nonnegative constant (representing the density of points), and $k=1$ in the case \LogU \ where we need to add one space dimension, and $0$ otherwise.
We say that  an electric field $E$ is compatible with $(\mc{C}, m)$ if it satisfies \eqref{eqclam}.  
By formulating in terms of electric fields, we have not retained the information that $E_N$ was a gradient.

\subsection{Definitions}
Let $\C$ be a point configuration, $m \geq 0$ and let $E$ be  compatible with $(\C, m)$. For any $\eta \in (0,1)$ we define the truncation of $E$ as
\begin{equation} \label{defEeta1}
E_{\eta}(x) := E(x) - \sum_{p \in \C} \nabla \f_{\eta}(x-p),
\end{equation}
where $\f_{\eta}$ is as in \eqref{def:truncation}.
Let us  observe that 
\be\label{divee}
-\div E_{\eta}= \cd\( \sum_{p\in \C} \delta_p^{(\eta)} - m\drd\).\ee This procedure is exactly the same, at the level of the electric fields, as the truncation procedure described in the previous sections. 

 In the sequel, $\carr_R$ still denotes the $\d$-dimensional cubes $[-R/2,R/2]^\d$.
\begin{defi}The (Coulomb) renormalized energy of $E$ with background $m$ is 
\begin{equation}\label{defW}
\mc{W}(E,m) := \lim_{\eta\to 0} \mc{W}_\eta(E,m)
\end{equation}
where we let
\begin{equation} \label{Weta}
\mc{W}_\eta(E,m) := \limsup_{R \ti} \frac{1}{R^{\d}} \int_{\carr_R\times \R^k}  |E_{\eta}|^2 - m \cd \g(\eta).
\end{equation}
\end{defi}
\begin{defi}
Let $\C$ be a point configuration and $m \geq 0$. We define the renormalized energy of $\C$ with background $m$ as
\begin{equation}\label{de522}
\W(\mc{C},m) := \inf\{\mc{W}(E,m) \ | \ E \  \text{compatible with } \ (\C, m) \}
\end{equation}
with the convention $\inf (\emptyset) = +\infty$.
\end{defi}

The name \textit{renormalized energy} (originating in Bethuel-Brezis-H\'{e}lein \cite{bbh} in the context of two-dimensional Ginzburg-Landau vortices) reflects the fact that the integral of $ |E|^2 $ is infinite, and is computed in a renormalized way by first applying a truncation and then removing the appropriate divergent part $m\cd \g(\eta)$.

It is not a priori clear how to define a  total Coulomb interaction of such a  jellium system, first because of the infinite size of the system as we just saw, second because of the lack of local charge neutrality of the system.
The definitions we presented avoid  having to go through computing  the sum of pairwise interactions between particles (it would not even  be clear how to sum them), but instead replace it with (renormalized variants of) the extensive quantity $\int |E|^2$  (see \eqref{formalcomputation} and the comments following it).

\begin{figure}[h!]
\begin{center}
\includegraphics[scale=0.7]{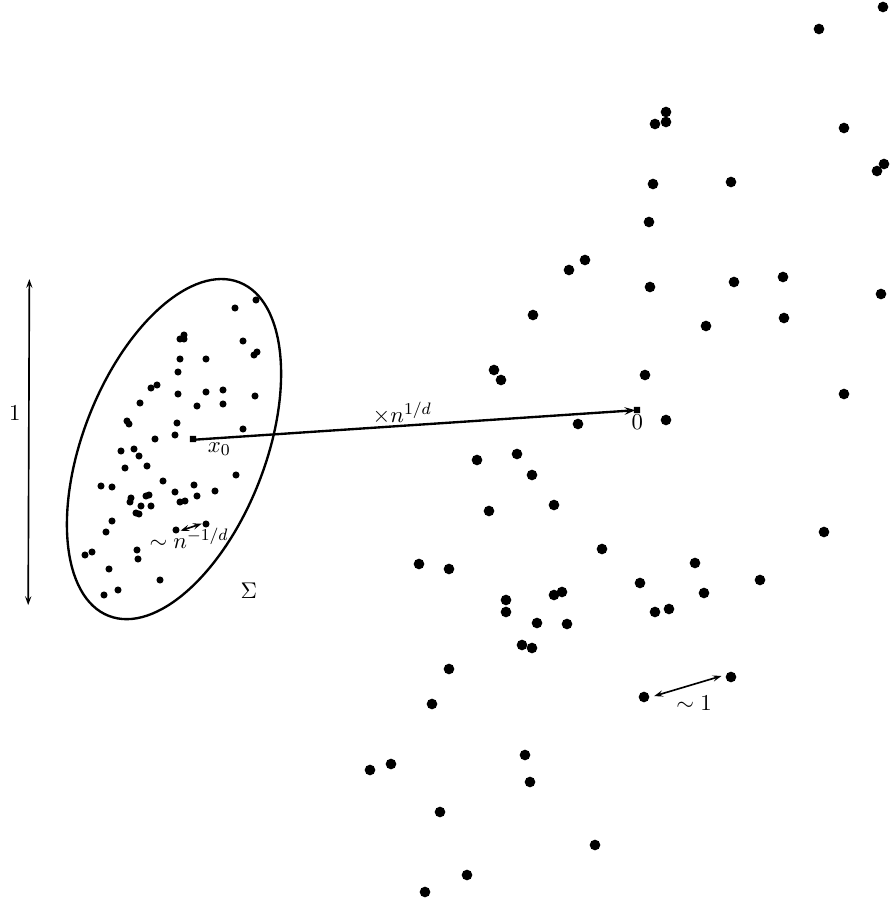}
\caption{An arbitrary blown-up configuration}\label{fig11}
\end{center}\end{figure}

\subsection{Scaling properties}
The dependence in $m$ can be scaled out as follows:
If $E\in \Elec_m$, we define $\sigma_m E$ by 
\begin{displaymath}
\sigma_m E := m^{1/\d - 1} E\left(\frac{\cdot}{m^{1/\d}} \right). \end{displaymath}
We have $\sigma_m E\in \Elec_1$ and  in the case \Coul,
\begin{equation} \label{scalingW}
\begin{cases}\mathcal{W} (E,m) = m ^{2-2/\d} \mathcal{W} (\sigma_m E,1) \\
\mathcal{W}_\eta (E,m) = m ^{2-2/\d} \mathcal{W}_{\eta m^{1/\d}}  (\sigma_mE,1) \end{cases}
\end{equation} and in the cases \LogU, \LogD
\be\label{scalingW2d}
\begin{cases}\mc{W}(E,m)= m\mc{W}(\sigma_m E,1) -\frac{2\pi}{\d}  m \log m \\
\mathcal{W}_\eta (E,m) = m \mathcal{W}_{\eta m^{1/\d}}  (\sigma_m E,1)-\frac{2 \pi}{\d} m  \log m  .
\end{cases}\ee
It also follows that 
\begin{equation}\label{scalingWW}
\left\{\begin{array}{ll}
 \W(\C,m)= m \W(\sigma_m \C,1)- \frac{2\pi}{\d} m \log m & \  \text{for} \, \LogU, \LogD\\
    \W(\C,m)= m^{2-2/\d}  \W(\sigma_m \C,1) &\ \text{for } \Coul .\end{array}\right.\end{equation}

In view of the above scaling relations, it suffices to study $\W(\cdot, 1) $.
It is proven in \cite{rs} that
 $\min \W(\cdot, 1)$ is finite and achieved for any $\d\ge 2$. Moreover, the minimum coincides with the limit as $R \rightarrow + \infty$ of the minima of $\W(\cdot, 1)$ on configurations that are $(R\mathbb{Z})^\d$-periodic (i.e. that live on the torus $\T_R = \R^\d / (R \mathbb{Z})^d$) (with $R\in \mathbb{N}$). 


\subsection{Partial results on the minimization of $\mathcal{W}$,  crystallization conjecture} \label{sec54}
We have seen  that the minima of $\W$  can be achieved as limits of the minima over periodic configurations (with respect to larger and larger tori). In the \LogU\  case, a convexity argument (for which we refer to \cite[Prop. 2.3]{ss2}) shows that the minimum is achieved when the points are equally spaced, in other words for the lattice or crystalline distribution $\mathbb{Z}$. 

In higher dimension, determining the value of  $\min \W$  is an open question. The only question that we can answer so far is that of the minimization over the restricted class of pure lattice configurations, in dimension $\d=2$ only,
i.e.   vector fields which are gradient of functions that are periodic with respect to a lattice $\mathbb{Z} \vec{u} + \mathbb{Z} \vec{v}$
 with $det(\vec{u}, \vec{v}) = 1$, corresponding to configurations of points that can be identified with $\mathbb{Z} \vec{u} + \mathbb{Z} \vec{v}$. In this case, we have :
\begin{theo}[The triangular lattice is the minimizer over lattices in 2D] \label{minimisationreseau} \mbox{}
The minimum of  $\W$ over this class of vector fields is achieved uniquely by the one corresponding to the  triangular ``Abrikosov"  lattice.
\end{theo}
Here the triangular lattice means $ \mathbb{Z} + \mathbb{Z} e^{i \pi/3}$, properly scaled, i.e. what is called the Abrikosov lattice in the context of superconductivity. This result is essentially a result of number theory, proven in the 50's.

One may ask  whether this triangular lattice does achieve the global minimum of $\W$.  The fact that the Abrikosov lattice is observed in superconductors, combined with the fact  that $\W$ can be derived  as 
 the limiting minimization problem of Ginzburg-Landau \cite{ssgl}, justify  conjecturing this.
It was recently proven in \cite{betermin} that this conjecture is equivalent to a conjecture of Brauchart-Hardin-Saff \cite{bhs} on the next order term in the asymptotic expansion of the minimal logarithmic energy on the sphere (an important problem in approximation theory, also related to Smale's ``7th problem for the 21st century"), which is obtained by  formal analytic continuation, hence by very different arguments. 

 In dimension $\d\ge 3$ the minimization of $\W$ even restricted to the class of lattices is   an open question, except in dimensions 8 and 24.
Similarly, one may conjecture that in low dimensions, the minimum of $\W$ is achieved by some particular lattice, see \cite{blanclewin}. In large dimensions, lattices are not expected to be minimizing.

\subsection{Renormalized energy for point processes}

Given $\XN = (x_1, \dots, x_N)$ in $(\R^{\d})^N$, we recall that we define $\XN'$ as the finite configuration rescaled by a factor $N^{1/{\d}}$ 
\begin{equation} \label{def:ompN}
\XN' := \sum_{i=1}^N \delta_{N^{1/{\d}} x_i}.
\end{equation}
It is a point configuration (an element of $\config$), which represents the $N$-tuple of particles $\XN$ seen at the microscopic scale.

The object we wish to describe is  the limit of the  \textit{tagged empirical field} $\bEmp_N[\XN]$ defined  as
\begin{equation}
\label{def:bEmp}
\bEmp_N[\XN] :=  \frac{1}{|\Sigma|} \int_{\Sigma} \delta_{\left(x,\,  \theta_{N^{1/{\d}} x} \cdot \XN' \right)} dx,
\end{equation}
where $\theta_x$ denotes the translation by $- x$. It is a probability measure on $\Sigma \times \config$.
For any $x \in \Sigma$, the term $\theta_{N^{1/{\d}}x} \cdot \XN'$ is an element of $\config$ which represents the $N$-tuple of particles $\XN$ centered at $x$ and seen at microscopic scale (or, equivalently, seen at microscopic scale and then centered at $N^{1/{\d}} x$). In particular any information about this point configuration in a given ball (around the origin) translates to an information about $\XN'$ around $x$. We may thus think of $\theta_{N^{1/{\d}}x} \cdot \XN'$ as encoding the behavior of $\XN'$ around $x$.

The limit of $\bEmp_N[\XN]$ will be a probability measure on $\Sigma \times \config$,  typically denoted $\bPst$,  whose first marginal is the (normalized) Lebesgue measure on $\Sigma$, and which is stationary with respect to the second variable. The space of such probability measures is denoted as $\mathcal{P}_s(\Sigma\times \config)$.  
We may then define the renormalized energy for $\bar P \in \mathcal{P}_s(\Sigma \times \config)$ as 
 \begin{equation}\bttW(\bar P,\mu)= \int_{\Sigma } \Esp_{\bar P^x} [ \W(\cdot, \mueq(x))]   \, dx,\end{equation}
where $\bar P^x$ is the disintegration with respect to the variable $x$ of the measure $\bar P$ (informally it is $\bar P(x, \cdot)$).
   \subsection{Lower bound for the energy in terms of the empirical field} 
We may now relate rigorously the energy $\HN$ to $\bttW$ by the following
  \begin{prop} \label{prop:LowerBoundenergies} 
  Assume $\pa \Sigma$ is $C^1$ and $\mu$ has a continuous  density on its support.
Let $\{\XN\}_N$ be a sequence of $N$-tuples of points in $\Rd$. If the left-hand side below is bounded independently of $N$, then up to extraction the sequence $\bEmp_N[\XN]$ converges to some $\bP$ in $\probas_s(\Sigma \times \config)$, and we have
\begin{equation} \label{gliminf}
\begin{cases} 
&\displaystyle \liminf_{N \ti} F_N^{\meseq}(\XN)+ \frac{N}{\d} \log N \ge \bttW(\bP, \mueq) \quad \text{in the cases \LogU, \LogD}\\
& \displaystyle \liminf_{N\ti} N^{1-\frac{2}{\d}} F_N^{\meseq}(\XN) \ge  \bttW(\bP, \mueq) \quad \text{in the cases \Coul}.\end{cases}
\end{equation}
\end{prop}
\begin{proof}
For each configuration $\XN$ let us define the tagged electric field process $P_N[\XN]$ (then we will drop the $\XN$) by 
\be \label{taggedef}
P_N[\XN](x,E):= \frac{1}{|\Sigma|}\int_{\Sigma} \delta_{(x,\theta_{N^{1/\d}x } \cdot E_N)}\, dx,\ee
where $E_N$ is as in \eqref{ele}.
The result follows from relatively standard considerations, after noting that by 
 Fubini's theorem we may write, for any $R>1$,
\begin{multline}\label{fubi}
\int_{\R^\d} |\nab {H_{N,\eta}^{\meseq}}'|^2 
\ge
N \int_{\R^\d}\frac{1}{|\carr_R|}\indic_{y\in \carr_R} \int_{\Sigma}  |E_{N,\eta} (N^{1/\d}x+y)|^2 \, dx \, dy \\ = N |\Sigma| \int f_{R,\eta}(E)\, dP_N(x,E)\end{multline}
 where  $f_{R,\eta}$ is defined by 
 $$f_{R,\eta}(E)=\frac{1}{|\carr_R|} \int_{\carr_R} |E_\eta|^2  $$ if $E$  is of the form $\nab H_N^\mu$ for some $N$-point configuration, and $+\infty$ otherwise.
 It then suffices to show that 
  $\{P_N\}_N$ is tight, and that up to extraction $P_N$ converges to some stationary tagged electric field process which can be ``projected down" to some $\bPst \in \probas_s(\Sigma\times \config)$, and take the limit  $N\to \infty$, $R\to \infty$, then $\eta\to 0$ in \eqref{fubi}.

\end{proof}

\frenchspacing

 \section{Large Deviations Principle for empirical fields}

We now wish to state a  large deviations result at the level of the point processes for the Gibbs measure \eqref{pnb2}, proven in \cite{lebles}.

\subsection{Specific relative entropy} \label{sec:defentropy}
We first  need to define the analogue of the entropy at the level of point processes: the specific relative entropy with respect to the Poisson point process. The Poisson point process with intensity $m$ is the point process characterized by the fact that for any bounded Borel set $B$ in $\R^\d$
$$P\( N(B)= n\)= \frac{(m|B|)^n}{n!}e^{-m|B|}$$
where $N(B)$ denotes the number of points in $B$. The expectation of the number of points in $B$ can then be computed to be $m|B|$, and one also observes that the number of points in two disjoint sets are independent, thus the points ``don't interact". 

For any $m \geq 0$, we denote by $\Poisson^m$ the (law of the) Poisson point process of intensity $m$ in $\Rd$ (it is an element of $\probas_s(\config)$). Let $P$ be in $\probas_s(\config)$. We define the specific relative entropy of $P$ with respect to $\Poisson^1$ as
\begin{equation} \label{def:ERS}
\ERS[P|\Poisson^1] := \lim_{R \ti} \frac{1}{R^{\d}} \ent[P_{\carr_R}|\Poisson^1_{\carr_R}],
\end{equation}
where $P_{\carr_R}, \Poisson^1_{\carr_R}$ denote the restrictions of the processes to the hypercube $\carr_R$. Here, $\ent[\cdot|\cdot]$ denotes the \textit{usual} relative entropy of two probability measures defined on the same probability space, namely
\begin{equation*}
\ent[\mu|\nu] := \int \log \frac{d\mu}{d\nu} \, d\mu
\end{equation*}
if $\mu$ is absolutely continuous with respect to $\nu$, and $+ \infty$ otherwise. 

\begin{lem} \label{lem:ERS} The following properties are known (cf. \cite{seppalainen}) :
\begin{enumerate}
\item The limit in \eqref{def:ERS} exists for $P$ stationary.
\item The map $P \mapsto \ERS[P |\Poisson^1]$ is affine and lower semi-continuous on $\probas_s(\config)$.
\item The sub-level sets of $\ERS[\cdot| \Poisson^1]$ are compact in $\probas_s(\config)$ (it is a \textit{good} rate function).
\item We have $\ERS[P| \Poisson^1] \geq 0$ and it vanishes only for $P = \Poisson^1$.
\end{enumerate}
\end{lem}

Next, if $\bar{P}$ is in $\probas_{s}(\Sigma \times \config)$, we defined the tagged  specific relative  entropy as
\begin{equation}
\label{def:bERS} \bERS[\bPst|\Poisson^1] := \int_{\Sigma} \ERS[\bPstx|\Poisson^1] dx.
\end{equation}



For any $N \geq 1$ we define $\QNbeta$ as
\begin{equation} \label{def:QNbeta}
d\QNbeta(x_1, \dots, x_N) := \prod_{i=1}^N \frac{\exp\left( - N^{\frac2\d} \beta \zeta(x_i) \right) dx_i}{\int_{\Rd} \exp\left( - N^{\frac2\d} \beta \zeta(x) \right) dx} , 
\end{equation} which exists by assumption,
and we let $\bQpN$ be the push-forward of $\QNbeta$ by the map $\bEmp_N$, as defined in \eqref{def:bEmp}.

We also introduce the following constant:
\begin{equation}
\label{defcomeg} \comeg := \log |\omega|-  |\Sigma|+1,
\end{equation}
where $\omega$ is the zero-set of $\zeta$ (see \eqref{defzeta}).

The following result expresses how the specific relative entropy evaluates the volume of ``microstates" $\XN$ near a reference tagged point process $\bPst$. It  is a Large Deviations Principle for the tagged empirical field, when the points are distributed according to a reference measure on $(\Rd)^N$ where there is no interaction. This is a microscopic or so-called ``type III" analogue of Sanov's theorem for the empirical measures. 

\begin{prop} \label{SanovbQN} 
For any $A \subset \probas_s(\Sigma \times \config)$, we have
\begin{multline} \label{EqSB}
- \inf_{\mathring{A} \cap \probas_{s,1}} \bERS[\bP|\Poisson^1] - \comeg \leq \liminf_{N \ti} \frac{1}{N} \log \bQpN (A) \\ \leq  \limsup_{N \ti} \frac{1}{N} \log \bQpN(A) \leq - \inf_{\bPst \in \bar{A}} \bERS[\bP|\Poisson^1] -\comeg.
\end{multline}
\end{prop}

With this result at hand we may expect the LDP result we are searching for  (or at least an LDP upper bound): indeed, the splitting \eqref{split0} and the rewriting \eqref{PNrewri} 
may be combined to Proposition \ref{prop:LowerBoundenergies}, which bounds from above the exponent, while Proposition \ref{SanovbQN} allows to evaluate the volume.
The hard part is to obtain the LDP lower bound, which involves instead an upper bound for the energy, together with a lower bound for the volume of configurations.

\subsection{Statement of the main result}

For any $\beta > 0$, we define a free energy functional $\fbarbeta$ as
\begin{equation}
\label{def:bfbeta} \fbarbeta(\bP) := \frac{\beta}{2} \bttW(\bP,\meseq) + \bERS[\bP|\Poisson^1].
\end{equation}

For any $N, \beta$ we let $\bPgot_{N,\beta}$ be push-forward of the canonical Gibbs measure $\PNbeta$ by the \textit{tagged empirical field map} $\bEmp_N$ as in \eqref{def:bEmp} (in other words, $\bPgot_{N, \beta}$ is the law of the tagged empirical field when the particles are distributed according to $\PNbeta$).

\begin{theo}[Large Deviation Principle for the tagged empirical fields \cite{lebles}] \label{TheoLDP} Under suitable assumptions on the regularity of $\partial \Sigma $ and $\meseq$,  for any $\beta>0$ the sequence $\{\bPgot_{N,\beta}\}_N$ satisfies a large deviation principle at speed $N$ with good rate function $\fbarbeta - \inf \fbarbeta$.
\end{theo}
In particular, in the limit $N\to \infty$, the law $\bPgot_{N,\beta}$ concentrates on minimizers of $\fbarbeta$. 
\begin{itemize}
\item 
One readily sees the effect of the temperature: in the minimization  there is a competition between the term $ \bttW$  based on the renormalized energy, which is expected to favor very ordered configurations, and the entropy term which in contrast favors disorder (the entropy is minimal for a Poisson point process).
  As $\beta \to \infty$, we expect configurations to concentrate on minimizers of $\bttW$ hence to crystallize on  lattices (in dimension 1, there is a complete proof \cite{leble}), in low dimensions. As $\beta\to 0$, the limiting point processes (if they exist) converge to the  Poisson point process. 
\item
This result is valid in the \LogU, \LogD, \Coul\  and also  \Riesz \ cases.
\item In the case \LogD,\ it is proven in \cite{loiloc} that the similar result holds down to the mesoscales.
\item
The corresponding, simpler, result for minimizers of $\HN$ (which are of interest on their own, as seen in Section \ref{intro}) can be proven as well. More precisely, the lower bound of Proposition \ref{prop:LowerBoundenergies} is sharp, and  the scaling property of $\bttW$ allow to obtain
\begin{multline}\label{formmin}\min \HN(\XN)= N^2 \I_V(\meseq) \\+ 
\begin{cases} &  - \frac{N}{\d}\log N+N \( -\frac{1}{2} \displaystyle \int_\Sigma \meseq \log \meseq + \frac{1}{\cd} \min \W(\cdot, 1)\) \ \text{in the cases }\LogU,\LogD \\ & N^{1+\frac{\s}{\d}} \frac{1}{\mathsf{c}_{\d,\s} }  \min \W(\cdot, 1) \displaystyle \int_{\Sigma}    \meseq^{1+\frac{\s}{\d}}  \ \  \text{in the cases  } \ \Coul, \Riesz.\end{cases}\end{multline} 
Moreover, if $\XN$ minimizes $\HN$ then, up to extraction $\bEmp_N[\XN]$ converges to a minimizer of $\bttW(\cdot, \meseq)$.  These are results of \cite{ss1,ss2,rs,ps}. More precise results including rigidity down to the microscopic scales and separation of points can be found in \cite{rns,PRN}.
\end{itemize}

As is usual and as alluded to before, as a by-product of the large deviation principle we obtain the order $N$ term in the expansion of the partition function.
\begin{coro}[Next-order expansion and thermodynamic limit]
\label{corothermo}
 Under the same assumptions, we have, as $N \to \infty$:
\begin{itemize}
\item In the logarithmic cases \LogU, \LogD
 \begin{equation} \label{expansionlog}
 \log \ZNbeta=  - \frac{\beta}{2} N^2 \mathcal \I_V(\meseq) +\frac{\beta}{2} \frac{N \log N}{\d} - N \min \fbarbeta +N(|\Sigma|-1) +o((\beta+1)N).
\end{equation}
\item In the  Coulomb cases \Coul
 \begin{equation}\label{expansionriesz}
 \log \ZNbeta= - \frac{\beta}{2} N^{1+\frac{2}{\d} } \mathcal \I_V(\meseq) - N \min \fbarbeta +N(|\Sigma|-1) +  o((\beta +1)N).
 \end{equation} 
\end{itemize} 
The scaling properties allow to rewrite the previous expansion in the  cases \LogU, \LogD\  as
\begin{multline} \label{logz}
 \log \ZNbeta =- \frac{\beta}{2} N^2 \mathcal \I_V(\meseq) +\frac{\beta}{2}\frac{N \log N}{\d} 
 -N C(\beta, \d)
 \\- N \left( 1 -\frac{\beta}{2 \d}\right) \int_\Sigma \meseq(x)\log \meseq(x) \, dx + o((\beta+1)N),
 \end{multline} 
where $C(\beta, \d)$ is a constant depending only on $\beta$ and the dimension $\d$, but \textit{independent of the potential}.
\end{coro}

In the particular case of \LogU \ with a quadratic potential $V(x) = x^2$, the equilibrium measure is known to be Wigner's semi-circular law and the limiting process at the microscopic scale when centered around a point $x \in (-2,2)$ (let us emphasize that here there is \textit{no averaging}) has been identified for any $\beta > 0$ in \cite{vv,kn}. It is called the \textit{sine-$\beta$ point process} and we denote it by $\sineb(x)$ (so that $\sineb(x)$ has intensity $\frac{1}{2\pi} \sqrt{4-x^2}$). 
For $\beta > 0$ fixed, the law of these processes do not depend on $x$ up to rescaling and we denote by $\sineb$ the corresponding process with intensity $1$. 

A corollary of our main result is then a new variational property of $\sineb$: it minimizes the sum of $\beta/2$ times the renormalized energy (with background $1$) and the specific relative entropy.

\subsection{Proof structure}
As mentioned above, combining  Propositions \ref{prop:LowerBoundenergies}  and \ref{SanovbQN}, we naturally obtain the LDP upper bound. 
Obtaining  the LDP lower bound is significantly harder, and is based on the following   \textit{quasi-continuity} of the interaction, in the following sense.

For any integer $N$, and $\delta > 0$ and any $\bP \in \probas_s(\Sigma \times \config)$, let us define
\begin{multline} \label{def:TNdelta}
T_N(\bP, \delta) := \Big\{ \XN \in (\Rd)^N, \\
 N^{(1-\frac{2}{\d})_+ }  F_N^{\meseq}(\XN)+ \(\frac{N}{\d} \log N \) \indic_{\LogU, \LogD} \leq \bttW(\bP, \meseq) + \delta \Big\}.
\end{multline}

\begin{prop} \label{quasicontinuite}
Let $\bPst\in\probas_{s, 1}(\Sigma \times \config)$. For  all $\delta_1, \delta_2>0$  we have 
\begin{equation}
\label{quasicontinu} 
\liminf_{N \ti} \frac{1}{N} \log \QNbeta \left( \{\bEmp_N \in B(\bPst, \delta_1)\} \cap T_N(\bP, \delta_2) \right) \ge  - \bERS[\bP|\Poisson^1]  -\comeg.
\end{equation}
\end{prop}
Let us compare with Proposition \ref{SanovbQN}. The first inequality of \eqref{EqSB} implies (taking $A = B(\bPst, \delta_1)$ and using the definition of $\bQpN$ as the push-forward of $\QNbeta$ by $\bEmp_N$) that
$$
\liminf_{N \ti} \frac{1}{N} \log \QNbeta \left( \{\bEmp_N \in B(\bPst, \delta_1)\} \right) \ge  - \bERS[\bP|\Poisson^1]  -\comeg.
$$
To obtain Proposition \ref{quasicontinuite}, which is the hard part of the proof, we thus need to show  that the event  $T_N(\bP, \delta_2)$ has  enough volume in phase-space near $\bPst$. This relies on the screening construction which we now  describe.

\subsection{Screening and consequences}

The screening procedure  allows, starting from a given configuration or set of configurations,  to build configurations whose energy is well controlled and can be computed additively in blocks of large microscopic size.  

Let us zoom everything by $N^{1/\d}$ and work in the rescaled coincidence set $\Sigma'= N^{1/\d} \Sigma$.  It will be split into cubes $K$ of size $R$ with $R$ independent of $N$ but large.  This can be done except in a thin boundary layer which needs to be treated separately (for simplicity we will not discuss this).

In a cube $K$  (or in $K \times \R$ in the \LogU\  case), the screening procedure  takes as input a given  configuration whose electric energy in the cube is not too large, and replaces it by a better configuration which is neutral (the number  of  points is equal to $\int_{K}{ \meseq}'$),
 and for which there is a compatible electric field in $K$ whose energy is controlled by the initial one  plus a negligible error as $R \to \infty$.  The point configuration is itself only modified in a thin layer near $\partial K$, and the electric field is constructed  
 in such a way as to obtain $E\cdot\vec\nu=0$ on $\pa  K$. The configuration is then said to be screened, because in some sense it then generates only a negligible electric field outside of $K$.
 

Starting from a given vector field, we not only construct a modified, screened, electric field, but a whole family of electric fields and configurations. This is important for the LDP lower bound where we need to exhibit a large enough volume of configurations having well-controlled energy.

Let us now get into more detail.

\subsubsection{Compatible and screened electric fields} \label{sec:notationscreen}
 Let  $K$ be a compact subset of $\R^{\d}$ with piecewise $C^1$ boundary (typically, $K$ will be an hyperrectangle or the support $\Sigma$), $\C$ be a finite point configuration in $K$, $\mu$ be a nonnegative density in $L^{\infty}(K)$, and $E$ be a vector field. 
\begin{itemize}
\item We say that $E$ is compatible with $(\C, \mu)$ in $K$  provided 
\begin{equation}
\label{def:compatible}
-\div  E= \cd \left(\mc{C} - \mu \right) \quad \text{in } K .
\end{equation}
\item We say that \textit{$E$ is compatible with $(\C, \mu)$ and screened in $K$}
when \eqref{def:compatible} holds and moreover
\begin{equation} \label{lescreening}
 E \cdot \vec{\nu} = 0   \quad \text{on } \partial  K, 
\end{equation}
where $\vec{\nu}$ is the outer unit normal vector. 
\end{itemize}
We note that by integrating \eqref{def:compatible} over $K$  and using \eqref{lescreening}, we find that $\int_K \C = \int_K \mu$. In other words, screened electric fields are necessarily associated with {\it neutral} systems, in which the positive charge $\int_K \C$ is equal to the negative charge $\int_K \mu$. In particular $\int_K \mu$ has to  be an integer. 

\subsubsection{Computing additively the energy}

Vector fields that are screened can be naturally ``pasted" together: 
 we may construct point configurations by elementary blocks (hyperrectangles) and compute their energy additively in these blocks.  More precisely, assume that space is partitioned into a family of hyperrectangles $K\in \mathcal K$ such that $\int_K d \meseq'$ are integers. We would like to  construct a vector field $\EK$ in each $K$ such that 
\begin{equation}
\label{eqcellfond}
\left\lbrace \begin{array}{ll}
 -\div \EK 
= \cd \left( \mc{C}_K - \meseq'\right) & \text{in} \ K \\
 \EK \cdot \vec{\nu} = 0  & \text{on} \  \partial  K \end{array}\right.
\end{equation} (where $\vec{\nu}$ is the outer unit normal to $K$)
for some discrete set of points $\mc{C}_K \subset K $, and with 
\begin{equation*}
\int_{K}  |\EK_\eta|^2 
\end{equation*}
 well-controlled.

When solving \eqref{eqcellfond}, we could take $\EK$ to be a gradient, but we do not require it, in fact  relaxing this constraint is what allows to paste together the vector fields. 
Once the relations \eqref{eqcellfond} are satisfied  on each $K$, we may paste together the vector fields $\EK$ into a unique vector field $\Etot$, and the discrete sets of points $\mc{C}_K$ into a configuration $\Ctot$. The cardinality of $\Ctot$ will be equal to $\int_{\R^{\d}} d{\meseq}'$, which is exactly $N$.

We may thus obtain a configuration of $N$ points, whose energy we may try to evaluate. The important fact is that the enforcement of the boundary condition $\EK\cdot\vec{ \nu}=0$ on each boundary ensures that 
\begin{equation}\label{deve}
-\div  \Etot= \cd \left(\Ctot - {\meseq}'\right) \quad \text{in } \R^{\d}
\end{equation}
holds globally. In general, a vector field which is discontinuous across an interface has a distributional divergence concentrated on the interface equal to the jump of the normal derivative, but here by construction the normal components coincide hence there is no divergence created across these interfaces. 

Even if the $\EK$'s were gradients, the global $\Etot$ is in general no longer a gradient.  This does not matter however, since the energy of the \textit{local electric field} $\nab \hpN$ generated by the finite configuration $\Ctot$ (and the background ${\muv}'$) is always smaller, as stated in 
\begin{lem} \label{minilocale} Let $\Omega$ be a compact subset of $\R^{\d}$ with piecewise $C^1$ boundary, let $N \geq 1$, let $\XN$ be in $(\Rd)^N$. We assume that all the points of $\XN$ belong to $K$ and that $ \Omega\subset \Sigma$. 
Let $E$ be a vector field  such that 
  \begin{equation}\label{checr}
\left\lbrace \begin{array}{ll}
 -\div  E
= \cd \left( \sum_{i=1}^N \delta_{x'_i} - {\meseq}' \right) & \text{in} \ \Omega \\
 E \cdot \vec{\nu} = 0  & \text{on} \  \partial  \Omega.
 \end{array}\right.
\end{equation} Then, for any $\eta \in (0,1)$ we have
\begin{equation}\label{comparloc}
\int_{\R^{\d}}  |\nabla H_{N,{\eta}}'|^2 \leq \int_{K}  |E_{\eta}|^2.
\end{equation}
\end{lem}
\begin{proof} The point is that $\nab {H_{N,\eta}^{\meseq}}'$ is the $L^2$ projection of any compatible $E$ onto gradients, and that the projection decreases the $L^2$ norm. The truncation procedure can be checked not to affect that property. 

\end{proof}

We thus have 
\begin{equation*}
\int_{\R^{\d}}  |\nab \hpNe|^2 \le \sum_{K\in \mathcal{K}} \int_{K} |\EK_\eta|^2,
\end{equation*}
and the energy $F_N^{\meseq}(\Ctot)$ can indeed be computed  additively over the cells (at least for an upper bound, which is precisely what we care about at this stage). This is the core of the method: by passing to the electric representation, we transform a sum of pairwise interactions into an additive (extensive) energy.

 The screening theorem states that this program can be achieved while adding only a negligible error to the energy, modifying only a small fraction of the points, and preserving the good separation. 
 It also allows, when starting from a family of configurations,  to build families of configurations  which occupy 
enough volume in configuration space (i.e.  losing only a negligible logarithmic volume). 


\subsection{Generating microstates and conclusion}
To prove Proposition \ref{quasicontinuite}, we generate configurations ``at random" whose volume is evaluated by Proposition \ref{SanovbQN}, and then modify them via the screening procedure.

  Since we want the global configurations to approximate (after averaging over translations) a given tagged point process $\bar{P}$, we  draw the point configuration in each hypercube $K\in \mathcal K$ jointly at random according to a Poisson point process, and standard large deviations results imply that enough of the averages ressemble $\bar{P}$.
Using Proposition \ref{SanovbQN}  they will still 
occupy a volume  bounded below by $ - \bERS[\bPst_{\um}|\Poisson^1]- \comeg $.

 These configurations then need to be modified by the screening procedure, while checking that these modifications do not alter much their phase-space volume, their energy, and keep them close to the given tagged process $\bar{P}$.
They also need to go through another modification, which we call {\it regularization} and which separates pairs of points that are too close (while decreasing the energy and not altering significantly the volume). This is meant to ensure that the error between $\bttW_\eta$ and $\bttW$ can be made small as $\eta \to 0$ uniformly among the configurations.

\end{document}